\documentclass[12pt]{article}

\usepackage{epsfig,amsmath,amssymb,amsfonts,amstext,amsthm}
\usepackage{latexsym,graphics,epsf,epsfig,color}
\usepackage{cite,url}

\newtheorem{theorem}{Theorem}

\newtheorem{lemma}{Lemma}
\newtheorem{coro}[lemma]{Corollary}

\newtheorem{fact}{Fact}

\begin{document}

\title{Generalized Cut-Set Bounds for Broadcast Networks}
\author{Amir~Salimi, Tie~Liu, and Shuguang~Cui
\thanks{This paper was presented in part at the 2012 International Symposium on Network Coding (NetCod), Cambridge, MA, June 2012. This research was supported in part by the Department of Defense under Grant HDTRA1-08-1-0010 and by the National Science Foundation under Grant CCF-08-45848. The authors are with the Department of Electrical and Computer Engineering, Texas A\&M University, College Station, TX 77843, USA (email: \{salimi,tieliu,cui\}@tamu.edu).}
}
 
\maketitle

\begin{abstract}
A broadcast network is a classical network with all source messages collocated at a single source node. For broadcast networks, the standard cut-set bounds, which are known to be loose in general, are closely related to union as a specific set operation to combine the basic cuts of the network. This paper provides a new set of network coding bounds for general broadcast networks. These bounds combine the basic cuts of the network via a variety of set operations (not just the union) and are established via only the submodularity of Shannon entropy. The tightness of these bounds are demonstrated via applications to combination networks.
\end{abstract}

\section{Introduction}
A \emph{classical network} is a capaciated directed acyclic graph $((V,A),(C_a:a\in A))$, where $V$ and $A$ are the node and the arc sets of the graph respectively, and $C_a$ is the link capacity for arc $a \in A$. A \emph{broadcast network} is a classical work for which all source messages are \emph{collocated} at a single source node. 

Consider a general broadcast network with one source node $s$ and $K$ sink nodes $t_k$, $k=1,\ldots,K$ (see Figure~\ref{fig:broadNet}). The source node $s$ has access to a collection of \emph{independent} messages $\mathsf{W}_I=(\mathsf{W}_i:i \in I)$, where $I$ is a finite index set. The messages intended for the sink node $t_k$ are given by $\mathsf{W}_{I_k}$, where $I_k$ is a nonempty subset of $I$. When all messages from $\mathsf{W}_I$ are \emph{unicast} messages, i.e., each of them is intended for \emph{only} one of the sink nodes, it follows from the celebrated max-flow min-cut theorem \cite{For-CJM56} that routing can achieve the entire capacity region of the network. On the other hand, when some of the messages from $\mathsf{W}_I$ are \emph{multicast} messages, i.e., they are intended for \emph{multiple} sink nodes, the capacity region of the network is generally \emph{unknown} except when there is only one multicast message at the source node \cite{Ahl-IT00,Li-IT03,Koe-ToN03} or there are only two sink nodes ($K=2$) in the network \cite{Ere-WCCC03,Nga-ICCCAS04,Ram-CWIT05}.

In this paper, we are interested in establishing strong network coding bounds for \emph{general} broadcast networks with multiple (multicast) messages and more than two sink nodes ($K \geq 3$). In particular, we are interested in network coding bounds that rely only on the \emph{cut} structure of the network. The rational behind this particular interest is two-folded. First, cut is a well-understood combinatorial structure for networks. Second, the fact that standard cut-set bounds \cite[Ch.~15.10]{Cov-B06} are \emph{tight} for the aforementioned special cases \cite{For-CJM56,Ahl-IT00,Li-IT03,Koe-ToN03,Ere-WCCC03,Nga-ICCCAS04,Ram-CWIT05} suggests that cut as a combinatorial structure can be useful for more general broadcast-network coding problems as well.

The starting point of this work is the following simple observation. For each $k =1,\ldots,K$, let $A_k$ be a ``basic" cut that separates the source node $s$ from the (single) sink node $t_k$. Then, for any nonempty subset $U \subseteq [K]:=\{1,\ldots,K\}$ the union $\cup_{k \in U}A_k$ is also a cut that separates the source node $s$ from the ``super" sink node $t_U$, whose intended messages are given by $\mathsf{W}_{\cup_{k \in U}I_k}$. By the standard cut-set bound \cite[Ch.~15.10]{Cov-B06}, we have
\begin{align}
R(\cup_{k \in U}I_k) \leq C(\cup_{k \in U}A_k)
\label{eq:CSB}
\end{align} 
for any achievable rate tuple $R_I:=(R_i:i \in I)$. Here, $R: 2^I \rightarrow \mathbb{R}^+$ is the \emph{rate} function that corresponds to the rate tuple $R_I$ and is given by
\begin{align}
R(I') := \sum_{i \in I'}R_i, \quad \forall I' \subseteq I,
\label{eq:rate}
\end{align}
and $C: 2^A \rightarrow \mathbb{R}^+$ is the \emph{capacity} function of the network where
\begin{align}
C(A') := \sum_{a \in A'}C_a, \quad \forall A' \subseteq A.
\label{eq:cap}
\end{align}

Note that the above observation depends critically on the fact that all messages $\mathsf{W}_I$ are \emph{collocated} at the source node $s$. When the messages are \emph{distributed} among several source nodes, it is well known that the union of several basic cuts may \emph{no longer} be a cut that separates the super source node from the super sink node and hence may not lead to any network coding bounds \cite{Kra-JNSM06}. 

\begin{figure}[!t]
\centering
\includegraphics[width=0.75\linewidth,draft=false]{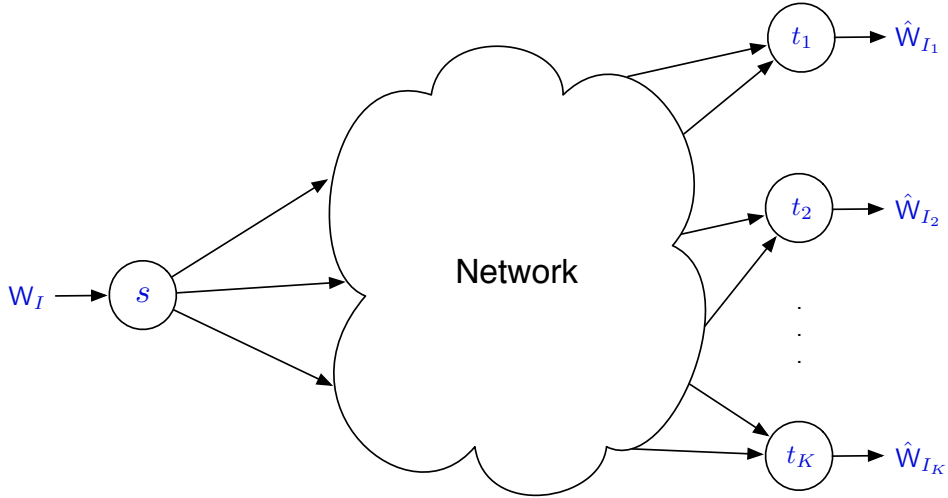}
\caption{Illustration of a general broadcast network.}
\label{fig:broadNet}
\end{figure}

Based on the above discussion, it is clear that for broadcast networks the \emph{standard} cut-set bounds \cite[Ch.~15.10]{Cov-B06} are closely related to \emph{union} as a specific set operation to combine different basic cuts of the network. Therefore, a natural question that one may ask is whether there are any other set operations (besides the union) that will also lead to nontrivial network coding bounds. 

In this paper, we provide a positive answer to the above question by establishing a new set of network coding bounds for general broadcast networks. We term these bounds \emph{generalized} cut-set bounds based on the facts that: 1) they rely only on the cut structure of the network; and 2) the set operations within the rate and the capacity functions are \emph{identical} (but not just the union any more), both similar to the case of standard cut-set bounds as in \eqref{eq:CSB}. From the proof viewpoint, as we shall see, these bounds are established via only the \emph{Shannon-type} inequalities. It is well known that all Shannon-type inequalities can be derived from the simple fact that Shannon entropy as a set function is \emph{submodular} \cite[Ch.~14.A]{Yeu-B08}. So, at heart, the generalized cut-set bounds are reflections of several new results that we establish on submodular function optimization.

The rest of the paper is organized as follows. In Section~\ref{sec:Mod} we establish several new results on submodular function optimization, which we shall use to prove the generalized cut-set bounds. A new set of network coding bounds that relate \emph{three} basic cuts of the network is provided in Section~\ref{sec:GCSB3}. The proof of these bounds is rather ``hands-on" and hence provides a good illustration on the essential idea on how to establish the generalized cut-set bounds. In Section~\ref{sec:GCSBK}, a new set of network coding bounds that relate arbitrary $K$ basic cuts of the network is provided, generalizing the bounds provided in Section~\ref{sec:GCSB3}. In Section~\ref{sec:CN}, the tightness of the generalized cut-set bounds is demonstrated via applications to \emph{combination} networks \cite{Nga-ITW04}. Finally, in Section~\ref{sec:Con} we conclude the paper with some remarks.

\section{Modular and Submodular Functions}\label{sec:Mod}
Let $S$ be a finite ground set. A function $f: 2^S \rightarrow \mathbb{R}^+$ is said to be \emph{submodular} if
\begin{align}
f(S_1)+f(S_2) &\geq f(S_1\cup S_2)+f(S_1\cap S_2), \quad \forall S_1,S_2 \subseteq S,
\label{eq:submod2}
\end{align}
and is said to be \emph{modular} if 
\begin{align}
f(S_1)+f(S_2) &= f(S_1\cup S_2)+f(S_1\cap S_2), \quad \forall S_1,S_2 \subseteq S.
\label{eq:mod2}
\end{align}

More generally, let $S_k$, $k=1,\ldots,K$, be a subset of $S$. For any nonempty subset $U$ of $[K]$ and any $r\in[|U|]$, let
\begin{align}
S^{(r)}(U) :=\cup_{\{U' \subseteq U: |U'|=r\}}\cap_{k \in U'}S_k.
\end{align}
Clearly, we have
\begin{align}
\cup_{k \in U}S_k = S^{(1)}(U) \supseteq S^{(2)}(U) \supseteq \cdots \supseteq S^{(|U|)}(U)=\cap_{k\in U}S_k
\label{eq:myOrd1}
\end{align}
for any nonempty $U \subseteq [K]$ and
\begin{align}
S^{(r)}(U') \subseteq S^{(r)}(U)
\label{eq:myOrd2}
\end{align}
for any $\emptyset \subset U' \subseteq U \subseteq [K]$ and any $r\in[|U'|]$. Furthermore, it is known that \cite[Th.~2]{Har-IT06}
\begin{align}
\sum_{k\in U}f(S_k) &\geq \sum_{r=1}^{|U|}f(S^{(r)}(U))
\label{eq:submodK}
\end{align}
if $f$ is a submodular function, and 
\begin{align}
\sum_{k\in U}f(S_k) &= \sum_{r=1}^{|U|}f(S^{(r)}(U))
\label{eq:modK}
\end{align}
if $f$ is a modular function.

Note that the standard submodularity \eqref{eq:submodK} relates $S^{(r)}(U)$ for different $r$ but a \emph{fixed} $U$. To establish the generalized cut-set bounds, however, we shall need the following technical results on modular and submodular functions that relate $S^{(r)}(U)$ for not only different $r$ but also \emph{different} $U$.

\begin{lemma}\label{lemma:1}
Let $r'$ and $J$ be two integers such that $0 \leq r' \leq J \leq K$. We have
\begin{align}
\sum_{r=1}^{r'}f(S_r)+\sum_{r=r'+1}^{J}f(S_r\cup S^{(r'+1)}([r])) & \ge
\sum_{r=1}^{r'}f(S^{(r)}([J]))+\sum_{r=r'+1}^{J}f(S^{(r'+1)}([r]))
\label{eq:1a}
\end{align}
if $f$ is a submodular function, and
\begin{align}
\sum_{r=1}^{r'}f(S_r)+\sum_{r=r'+1}^{J}f(S_r\cup S^{(r'+1)}([r])) & =
\sum_{r=1}^{r'}f(S^{(r)}([J]))+\sum_{r=r'+1}^{J}f(S^{(r'+1)}([r]))
\end{align}
if $f$ is a modular function. 
\end{lemma}

Note that when $r'=0$, we have $S^{(r'+1)}([r]) = S^{(1)}([r]) = \cup_{k=1}^{r}S_k \supseteq S_r$ for any $r=1,\ldots,J$. In this case, the inequality \eqref{eq:1a} reduces to the 
trivial equality
\begin{align}
\sum_{r=1}^{J}f(S^{(1)}([r])) &= \sum_{r=1}^{J}f(S^{(1)}([r])).
\end{align}
On the other hand, when $r'=J$, the inequality \eqref{eq:1a} reduces to the standard submodularity
\begin{align}
\sum_{r=1}^{J}f(S_r) & \geq \sum_{r=1}^{J}f(S^{(r)}([J])).
\end{align}
For the general case where $0 < r' < J$, a proof of the lemma is provided in Appendix~\ref{app:pf-lemma1}.

Let $S'_k:=S_k\cup S_0$ for $k=1,\ldots,K$. For any nonempty $U \subseteq [K]$ and any $r=1,\ldots,|U|$ we have
\begin{align}
S'^{(r)}(U) 
&=\cup_{\{U' \subseteq U: |U'|=r\}}\cap_{k \in U'}S'_k\\
&=\cup_{\{U' \subseteq U: |U'|=r\}}\cap_{k \in U'}(S_k\cup S_0)\\
& = \left(\cup_{\{U' \subseteq U: |U'|=r\}}\cap_{k \in U'}S_k\right)\cup S_0\\
&= S^{(r)}(U) \cup S_0.\label{eq:1c}
\end{align}
Applying Lemma~\ref{lemma:1} for $S'_k$, $k=1,\ldots,K$, and \eqref{eq:1c}, we have the following corollary.

\begin{coro}\label{cor:1}
Let $r'$ and $J$ be two integers such that $0 \leq r' \leq J \leq K$, and let $S_0$ be a subset of $S$. We have
\begin{align}
\sum_{r=1}^{r'}f(S_r\cup S_0)&+\sum_{r=r'+1}^{J}f(S_r\cup S^{(r'+1)}([r])\cup S_0) \notag\\
& \ge \sum_{r=1}^{r'}f(S^{(r)}([J])\cup S_0)+\sum_{r=r'+1}^{J}f(S^{(r'+1)}([r])\cup S_0)
\label{eq:2a}
\end{align}
if $f$ is a submodular function, and
\begin{align}
\sum_{r=1}^{r'}f(S_r\cup S_0)&+\sum_{r=r'+1}^{J}f(S_r\cup S^{(r'+1)}([r])\cup S_0)\notag\\
& = \sum_{r=1}^{r'}f(S^{(r)}([J])\cup S_0)+\sum_{r=r'+1}^{J}f(S^{(r'+1)}([r])\cup S_0)
\end{align}
if $f$ is a modular function. 
\end{coro}

We shall also need the following lemma, for which a proof is provided in Appendix~\ref{app:pf-lemma2}.

\begin{lemma}\label{lemma:2}
Let $U$ and $T$ be two nonempty subsets of $[K]$. Write, without loss of generality, that $T=\{t_1,\ldots,t_{|T|}\}$ where $1 \leq t_1 < t_2 < \cdots < t_{|T|} \leq K$. Let $q$ and $r_q$ be two integers such that $1 \leq q \leq |U|$, $1 \leq r_q \leq |T|$, and $S^{(q)}(U) \subseteq S^{(r_q)}(T)$. We have
\begin{align}
\sum_{r=1}^{|T|}&f(S_{t_r})+r_qf(S^{(q)}(U))\notag\\
& \geq \sum_{r=1}^{r_q}\left(f(S^{(r)}(T))+f(S_{t_r}\cap S^{(q)}(U))\right)+\sum_{r=r_q+1}^{|T|}f(S_{t_r}\cap(S^{(q)}(U)\cup S^{(r_q+1)}(\{t_1,\ldots,t_r\})))
\label{eq:3a}
\end{align}
if $f$ is a submodular function, and
\begin{align}
\sum_{r=1}^{|T|}&f(S_{t_r})+r_qf(S^{(q)}(U))\notag\\
& = \sum_{r=1}^{r_q}\left(f(S^{(r)}(T))+f(S_{t_r}\cap S^{(q)}(U))\right)+\sum_{r=r_q+1}^{|T|}f(S_{t_r}\cap(S^{(q)}(U)\cup S^{(r_q+1)}(\{t_1,\ldots,t_r\})))\label{eq:3b}
\end{align}
if $f$ is a modular function.
\end{lemma}

For specific functions, let $\mathsf{Z}_S:=(\mathsf{Z}_i:i \in S)$ be a collection of jointly distributed random variables, and let $H(\mathsf{Z}_S)$ be the joint (Shannon) entropy of $\mathsf{Z}_S$. Then, it is well known \cite[Ch.~14.A]{Yeu-B08} that $H_\mathsf{Z}: 2^S \rightarrow \mathbb{R}^+$ where 
\begin{align}
H_{\mathsf{Z}}(S') := H(\mathsf{Z}_{S'}), \quad \forall S' \subseteq S
\end{align}
is a \emph{submodular} function. Furthermore, it is straightforward to verify that the rate function $R(\cdot)$ (for a given rate tuple $R_I$) and the capacity function $C(\cdot)$, defined in \eqref{eq:rate} and \eqref{eq:cap} respectively, are \emph{modular} functions.

\section{Generalized Cut-Set Bounds Relating Three Basic Cuts of the Network}\label{sec:GCSB3}
\subsection{Main Result}
\begin{theorem}\label{thm:GCSB3}
Consider a broadcast network with a collection of independent messages $\mathsf{W}_I$ collocated at the source node $s$ and $K \geq 3$ sink nodes $t_k$, $k=1,\ldots,K$. For any $k=1,\ldots,K$, let $\mathsf{W}_{I_k}$ be the intended messages for the sink node $t_k$, and let $A_k$ be a basic cut that separates the source node $s$ from the sink node $t_k$. We have
\begin{align}
R(I_i \cup I_j \cup I_k)+R&(I_i \cap I_j) \leq C(A_i \cup A_j \cup A_k)+C(A_i \cap A_j),\label{eq:GCSB3a}\\
R(I_i \cup I_j \cup I_k)+R&((I_i \cap I_j) \cup (I_i \cap I_k) \cup (I_j \cap I_k))\notag\\
&\leq C(A_i\cup A_i \cup A_k)+C((A_i \cap A_j) \cup (A_i \cap A_k) \cup (A_j \cap A_k)),\label{eq:GCSB3b}\\
R(I_i \cup I_j \cup I_k)+R&(I_i \cup I_j) +R(I_i \cap I_j \cap I_k)\notag\\
&\leq C(A_i \cup A_j \cup A_k)+C(A_i \cup A_j)+C(A_i\cap A_j \cap A_k),\label{eq:GCSB3c}\\
\mbox{and} \quad 2R(I_i \cup I_j \cup I_k)+R&(I_i \cap I_j \cap I_k) \leq 2C(A_i\cup A_j \cup A_k)+C(A_i \cap A_j \cap A_k)\label{eq:GCSB3d}
\end{align}
for any achievable rate tuple $R_I$ and any three distinct integers $i$, $j$, and $k$ from $[K]$.
\end{theorem}

\begin{figure}[!t]
\centering
\includegraphics[width=0.8\linewidth,draft=false]{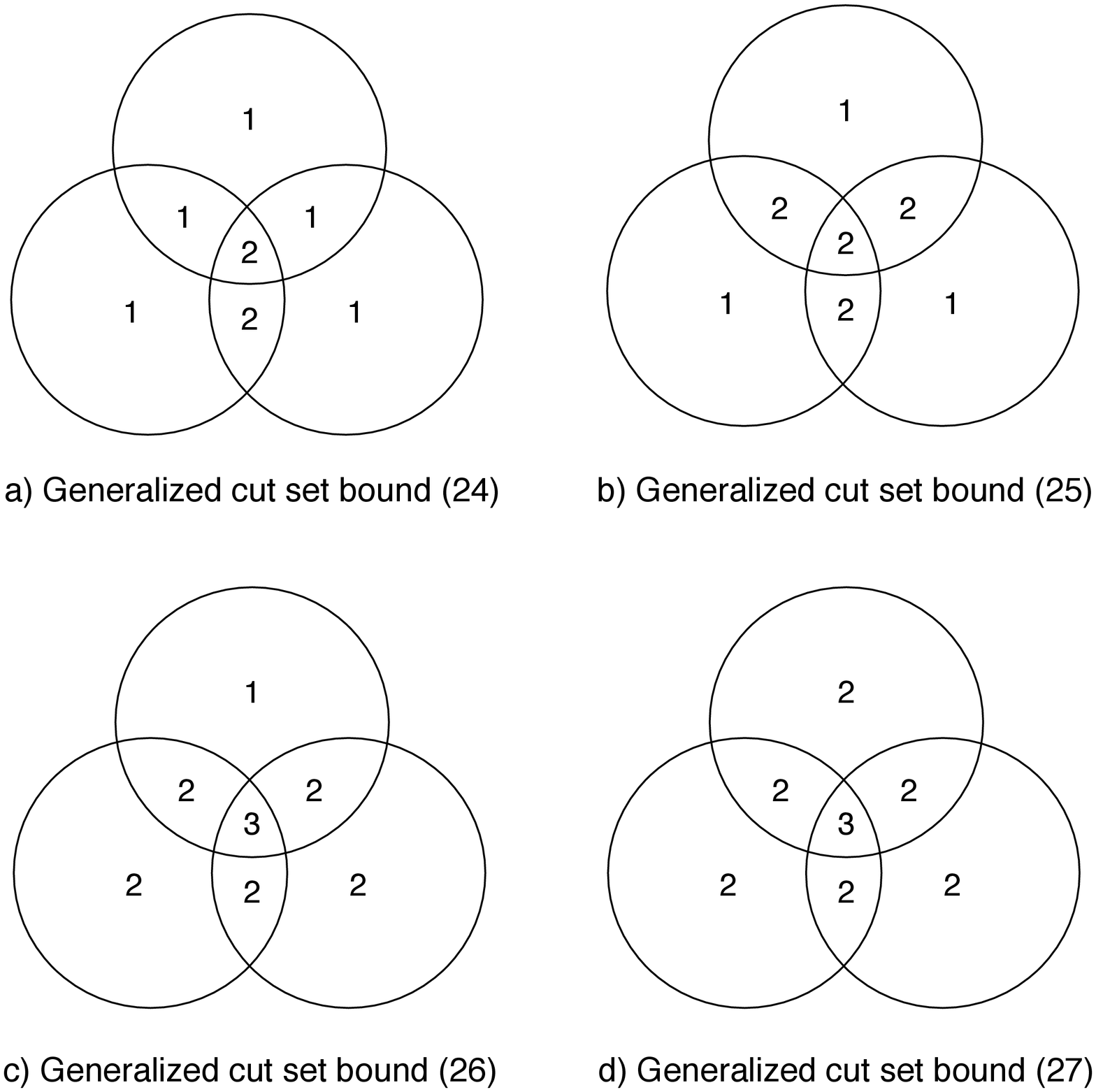}
\caption{The weight distributions for the generalized cut-set bounds \eqref{eq:GCSB3a}--\eqref{eq:GCSB3d}. Here, each circle represents the set of the messages intended for a particular sink node. The number within each separate area indicates the weight for the rates of the messages represented by the area.}
\label{fig:WD}
\end{figure}

Note that the left-hand sides of the generalized cut-set bounds \eqref{eq:GCSB3a}--\eqref{eq:GCSB3d} are \emph{weighted} sum rates with integer weights on the rates of the messages from $\mathsf{W}_{I_i \cup I_j \cup I_k}$. Figure~\ref{fig:WD} illustrates the weight distributions for the generalized cut-set bounds \eqref{eq:GCSB3a}--\eqref{eq:GCSB3d}.

\subsection{Proof of Theorem~\ref{thm:GCSB3}}
Let $(n,\{\mathsf{X}_a: a \in A\})$ be an \emph{admissible} code with block length $n$, where $\mathsf{X}_a$ is the message transmitted over the arc $a$. By the independence bound \cite[Th~2.6.6]{Cov-B06} and the link-capacity constraints, we have
\begin{align}
H_{\mathsf{X}}(A') \leq \sum_{a \in A'}H(\mathsf{X}_a) \leq n\sum_{a \in A'}C_a = nC(A'), \quad \forall A' \subseteq A.
\label{eq:linkCap}
\end{align}
For notational simplicity, in this proof we shall assume \emph{perfect} recovery of the messages at each of the sink nodes. It should be clear from the proof that by applying the well-known Fano's inequality \cite[Th~2.10.1]{Cov-B06}, the results also hold for \emph{asymptotically perfect} recovery. By the perfect recovery requirement, for any nonempty subset $U \subseteq [K]$ the collection of the messages $\mathsf{W}_{\cup_{k \in U}I_k}$ must be a \emph{function} of the messages $\mathsf{X}_{\cup_{k \in U}A_k}$ transmitted over the $s$-$t_U$ cut $\cup_{k \in U}A_k$. We thus have
\begin{align}
H_{\mathsf{W}}(\cup_{k \in U}I_k) \leq H_{\mathsf{X}}(\cup_{k \in U}A_k), \quad \forall U \subseteq [K].
\label{eq:perfectDec}
\end{align}

\noindent \emph{Proof of \eqref{eq:GCSB3a}.} Let $U=\{i,j,k\}$ in \eqref{eq:perfectDec}. Denote by 
\begin{align}
I_{\mathsf{X}}(A_i;A_j) := I(\mathsf{X}_{A_i};\mathsf{X}_{A_j})
\end{align}
the mutual information between $\mathsf{X}_{A_i}$ and $\mathsf{X}_{A_j}$. We have
\begin{align}
H_{\mathsf{W}}(I_i\cup I_j \cup I_k) & \leq H_{\mathsf{X}}(A_i\cup A_j \cup A_k)\label{pf3:10}\\
& = H_{\mathsf{X}}(A_i)+H_{\mathsf{X}}(A_j|A_i)+H_{\mathsf{X}}(A_k|A_i\cup A_j)\label{pf3:20}\\
& = H_{\mathsf{X}}(A_i)+\left(H_{\mathsf{X}}(A_j)-I_{\mathsf{X}}(A_i;A_j)\right)+\left(H_{\mathsf{X}}(A_k)-I_{\mathsf{X}}(A_k;A_i\cup A_j)\right)\label{pf3:30}\\
& = H_{\mathsf{X}}(A_i)+\left(H_{\mathsf{X}}(A_j)-I_{\mathsf{X},\mathsf{W}}(A_i,I_i;A_j,I_j)\right)+\notag\\
& \hspace{15pt} \left(H_{\mathsf{X}}(A_k)-I_{\mathsf{X},\mathsf{W}}(A_k,I_k;A_i\cup A_j,I_i\cup I_j)\right)\label{pf3:40}\\
& \leq H_{\mathsf{X}}(A_i)+\left(H_{\mathsf{X}}(A_j)-H_{\mathsf{X},\mathsf{W}}(A_i\cap A_j,I_i\cap I_j)\right)+\notag\\
& \hspace{15pt} \left(H_{\mathsf{X}}(A_k)-H_{\mathsf{X},\mathsf{W}}(A_k\cap(A_i\cup A_j),I_k\cap(I_i\cup I_j))\right)\label{pf3:50}
\end{align}
where \eqref{pf3:40} follows from the fact that: 1) $\mathsf{W}_{I_i}$ and $\mathsf{W}_{I_j}$ are functions of $\mathsf{X}_{A_i}$ and $\mathsf{X}_{A_j}$ respectively so we have $I_{\mathsf{X}}(A_i;A_j)=I_{\mathsf{X},\mathsf{W}}(A_i,I_i;A_j,I_j)$; and 2) $\mathsf{W}_{I_k}$ and $\mathsf{W}_{I_i\cup I_j}$ are functions of $\mathsf{X}_{A_k}$ and $\mathsf{X}_{A_i\cup A_j}$ respectively so we have $I_{\mathsf{X}}(A_k;A_i\cup A_j)=I_{\mathsf{X},\mathsf{W}}(A_k,I_k;A_i\cup A_j,I_i\cup I_j)$, and \eqref{pf3:50} follows from the fact that
\begin{align}
I_{\mathsf{X},\mathsf{W}}(A_i,I_i;A_j,I_j) & \geq I_{\mathsf{X},\mathsf{W}}(A_i\cap A_j,I_i\cap I_j;A_i\cap A_j,I_i\cap I_j)\\
& = H_{\mathsf{X},\mathsf{W}}(A_i\cap A_j,I_i\cap I_j)
\end{align}
and
\begin{align}
I_{\mathsf{X},\mathsf{W}}&(A_k,I_k;A_i\cup A_j,I_i\cup I_j)\notag\\
& \geq I_{\mathsf{X},\mathsf{W}}(A_k\cap(A_i\cup A_j),I_k\cap(I_i\cup I_j);A_k\cap(A_i\cup A_j),I_k\cap(I_i\cup I_j))\\
&=H_{\mathsf{X},\mathsf{W}}(A_k\cap(A_i\cup A_j),I_k\cap(I_i\cup I_j)).
\end{align}
Note that we trivially have
\begin{align}
H_{\mathsf{X},\mathsf{W}}(A_i\cap A_j,I_i\cap I_j) & \geq H_{\mathsf{W}}(I_i\cap I_j) \label{pf3:60}\\
\mbox{and} \quad H_{\mathsf{X},\mathsf{W}}(A_k\cap(A_i\cup A_j),I_k\cap(I_i\cup I_j)) & \geq 
H_{\mathsf{X}}(A_k\cap(A_i\cup A_j)). \label{pf3:70}
\end{align}
Substituting \eqref{pf3:60} and \eqref{pf3:70} into \eqref{pf3:50} gives
\begin{align}
H_{\mathsf{W}}(I_i\cup I_j \cup I_k)+H_{\mathsf{W}}(I_i\cap I_j) & \leq
H_{\mathsf{X}}(A_i)+H_{\mathsf{X}}(A_j)+H_{\mathsf{X}}(A_k)-H_{\mathsf{X}}(A_k\cap(A_i\cup A_j))\label{pf3:80}\\
& \leq H_{\mathsf{X}}(A_i)+H_{\mathsf{X}}(A_j)+H_{\mathsf{X}}(A_k\setminus(A_i\cup A_j))\label{pf3:90}\\
& \leq n\left(C(A_i)+C(A_j)+C(A_k\setminus(A_i\cup A_j))\right)\label{pf3:100}\\
& = n\left(C(A_i\cup A_j \cup A_k)+C(A_i\cap A_j)\right)\label{pf3:110}
\end{align}
where \eqref{pf3:90} follows from the independence bound
\begin{align}
H_{\mathsf{X}}(A_k) \leq H_{\mathsf{X}}(A_k\cap(A_i\cup A_j))+
H_{\mathsf{X}}(A_k\setminus(A_i\cup A_j));\label{pf3:115}
\end{align}
\eqref{pf3:100} follows from \eqref{eq:linkCap} for $A'=A_i$, $A_j$, and $A_k\setminus(A_i\cup A_j)$; and \eqref{pf3:110} follows from the fact that the capacity function $C(\cdot)$ is a modular function. Substituting 
\begin{align}
H_{\mathsf{W}}(I_i\cup I_j \cup I_k) & = nR(I_i\cup I_j \cup I_k)\label{pf3:116}\\
\mbox{and} \quad H_{\mathsf{W}}(I_i\cap I_j) & =R(I_i\cap I_j)\label{pf3:117}
\end{align}
into \eqref{pf3:110} and dividing both sides of the inequality by $n$ complete the proof of \eqref{eq:GCSB3a}. \hfill $\square$

We note here that if we had directly bounded from above the right-hand side of \eqref{pf3:10} by $nC(A_i\cup A_j \cup A_k)$ using the independence bound, it would have led to the standard cut-set bound
\begin{align}
R(I_i\cup I_j \cup I_k) \leq C(A_i\cup A_j \cup A_k).
\end{align}
But the use of the independence bound would have implied that all messages transmitted over $A_i\cup A_j \cup A_k$ are \emph{independent}, which may not be the case in the presence of multicast messages.

\noindent \emph{Proof of \eqref{eq:GCSB3b}.} Applying the two-way submodularity \eqref{eq:submod2} of the Shannon entropy with $\mathsf{Z}=(\mathsf{X},\mathsf{W})$, $S_1 = (A_i\cap A_j,I_i\cap I_j)$, and $S_2 = (A_k\cap(A_i\cup A_j),I_k\cap(I_i\cup I_j))$, we have
\begin{align}
H_{\mathsf{X},\mathsf{W}}&(A_i\cap A_j,I_i\cap I_j)+H_{\mathsf{X},\mathsf{W}}(A_k\cap(A_i\cup A_j),I_k\cap(I_i\cup I_j))\notag\\
&\geq H_{\mathsf{X},\mathsf{W}}(A_i\cap A_j\cap A_k,I_i\cap I_j\cap I_k)+\notag\\
&\hspace{15pt} H_{\mathsf{X},\mathsf{W}}((A_i\cap A_j)\cup(A_i\cap A_k)\cup(A_j\cap A_k),(I_i\cap I_j)\cap(I_i\cup I_k)\cap(I_j\cup I_k))\label{pf3:120}\\
&\geq H_{\mathsf{X}}(A_i\cap A_j\cap A_k)+
H_{\mathsf{W}}((I_i\cap I_j)\cup(I_i\cap I_k)\cup(I_j\cap I_k)).\label{pf3:130}
\end{align}
Substituting \eqref{pf3:130} into \eqref{pf3:50} gives
\begin{align}
H_{\mathsf{W}}&(I_i\cup I_j \cup I_k)+H_{\mathsf{W}}((I_i\cap I_j)\cup(I_j\cap I_k)\cup(I_k\cap I_i))\notag\\
& \leq H_{\mathsf{X}}(A_i)+H_{\mathsf{X}}(A_j)+H_{\mathsf{X}}(A_k)-H_{\mathsf{X}}(A_i\cap A_j\cap A_k)\label{pf3:140}\\
& \leq H_{\mathsf{X}}(A_i)+H_{\mathsf{X}}(A_j)+H_{\mathsf{X}}(A_k\setminus(A_i\cap A_j))\label{pf3:150}\\
& \leq n\left(C(A_i)+C(A_j)+C(A_k\setminus(A_i\cap A_j))\right)\label{pf3:160}\\
& = n\left(C(A_i\cup A_j \cup A_k)+C((A_i\cap A_j)\cup(A_i\cap A_k)\cup(A_j\cap A_k))\right)\label{pf3:170}
\end{align}
where \eqref{pf3:150} follows from the independence bound
\begin{align}
H_{\mathsf{X}}(A_k) \leq H_{\mathsf{X}}(A_k\cap(A_i\cap A_j))+
H_{\mathsf{X}}(A_k\setminus(A_i\cap A_j));
\end{align}
\eqref{pf3:160} follows from \eqref{eq:linkCap} for $A'=A_i$, $A_j$, and $A_k\setminus(A_i\cap A_j)$; and \eqref{pf3:170} follows from the fact that the capacity function $C(\cdot)$ is a modular function. Substituting \eqref{pf3:116} and
\begin{align}
H_{\mathsf{W}}((I_i\cap I_j)\cup(I_i\cap I_k)\cup(I_j\cap I_k)) &= nR((I_i\cap I_j)\cup(I_i\cap I_k)\cup(I_j\cap I_k))
\end{align}
into \eqref{pf3:170} and dividing both sides of the inequality by $n$ complete the proof of \eqref{eq:GCSB3b}. \hfill $\square$

\noindent \emph{Proof of \eqref{eq:GCSB3c}.} By the \emph{symmetry} among $i$, $j$, and $k$ in \eqref{pf3:50}, we have
\begin{align}
H_{\mathsf{W}}(I_i\cup I_j \cup I_k) & \leq H_{\mathsf{X}}(A_i)+\left(H_{\mathsf{X}}(A_k)-H_{\mathsf{X},\mathsf{W}}(A_i\cap A_k,I_i\cap I_k)\right)+\notag\\
& \hspace{15pt} \left(H_{\mathsf{X}}(A_j)-H_{\mathsf{X},\mathsf{W}}(A_j\cap(A_i\cup A_k),I_j\cap(I_i\cup I_k))\right).\label{pf3:180}
\end{align}
Also note that
\begin{align}
H_{\mathsf{W}}(I_i\cup I_j) & \leq H_{\mathsf{X}}(A_i\cup A_j)\label{pf3:190}\\
& = H_{\mathsf{X}}(A_i)+H_{\mathsf{X}}(A_j|A_i)\label{pf3:200}\\
& = H_{\mathsf{X}}(A_i)+\left(H_{\mathsf{X}}(A_j)-I_{\mathsf{X}}(A_i;A_j)\right)\label{pf3:210}\\
& = H_{\mathsf{X}}(A_i)+\left(H_{\mathsf{X}}(A_j)-I_{\mathsf{X},\mathsf{W}}(A_i,I_i;A_j,I_j)\right)\label{pf3:220}\\
& \leq H_{\mathsf{X}}(A_i)+\left(H_{\mathsf{X}}(A_j)-H_{\mathsf{X},\mathsf{W}}(A_i\cap A_j,I_i\cap I_j)\right).\label{pf3:230}
\end{align}
Adding \eqref{pf3:180} and \eqref{pf3:230} gives
\begin{align}
H_{\mathsf{W}}&(I_i\cup I_j \cup I_k)+H_{\mathsf{W}}(I_i\cup I_j)\notag\\
& \leq 2H_{\mathsf{X}}(A_i)+2H_{\mathsf{X}}(A_j)+H_{\mathsf{X}}(A_k)-H_{\mathsf{X},\mathsf{W}}(A_i\cap A_j,I_i\cap I_j)-\notag\\
& \hspace{15pt} H_{\mathsf{X},\mathsf{W}}(A_i\cap A_k,I_i\cap I_k)-H_{\mathsf{X},\mathsf{W}}(A_j\cap(A_i\cup A_k),I_j\cap(I_i\cup I_k)).\label{pf3:240}
\end{align}
Applying the two-way submodularity \eqref{eq:submod2} of the Shannon entropy with $\mathsf{Z}=(\mathsf{X},\mathsf{W})$, $S_1 = (A_i\cap A_j,I_i\cap I_j)$, and $S_2 = (A_i\cap A_k,I_i\cap I_k)$, we have
\begin{align}
H_{\mathsf{X},\mathsf{W}}&(A_i\cap A_j,I_i\cap I_j)+H_{\mathsf{X},\mathsf{W}}(A_i\cap A_k,I_i\cap I_k)\notag\\
&\geq H_{\mathsf{X},\mathsf{W}}(A_i\cap A_j\cap A_k,I_i\cap I_j\cap I_k)+H_{\mathsf{X},\mathsf{W}}(A_i\cap(A_j\cup A_k),I_i\cap(I_j\cup I_k))\\
&\geq H_{\mathsf{W}}(I_i\cap I_j\cap I_k)+H_{\mathsf{X}}(A_i\cap(A_j\cup A_k)).\label{pf3:260}
\end{align}
Note that we trivially have
\begin{align}
H_{\mathsf{X},\mathsf{W}}(A_j\cap(A_i\cup A_k),I_j\cap(I_i\cup I_k)) \geq H_{\mathsf{X}}(A_j\cap(A_i\cup A_k)).
\label{pf3:265}
\end{align}
Substituting \eqref{pf3:260} and \eqref{pf3:265} into \eqref{pf3:240}, we have
\begin{align}
H_{\mathsf{W}}&(I_i\cup I_j \cup I_k)+H_{\mathsf{W}}(I_i\cup I_j)+H_{\mathsf{W}}(I_i\cap I_j \cap I_k)\notag\\
& \leq 2H_{\mathsf{X}}(A_i)+2H_{\mathsf{X}}(A_j)+H_{\mathsf{X}}(A_k)-H_{\mathsf{X}}(A_i\cap(A_j\cup A_k))-H_{\mathsf{X}}(A_j\cap(A_i\cup A_k))\\
& \leq H_{\mathsf{X}}(A_i)+H_{\mathsf{X}}(A_j)+H_{\mathsf{X}}(A_k)+H_{\mathsf{X}}(A_i\setminus(A_j\cup A_k))+H_{\mathsf{X}}(A_j\setminus(A_i\cup A_k))\label{pf3:270}\\
& \leq n\left(C(A_i)+C(A_j)+C(A_k)+C(A_i\setminus(A_j\cup A_k))+C(A_j\setminus(A_i\cup A_k))\right)\label{pf3:280}\\
& = n\left(C(A_i\cup A_j\cup A_k)+C(A_i\cup A_j)+C(A_i\cap A_j\cap A_k)\right)\label{pf3:290}
\end{align}
where \eqref{pf3:270} follows from the independence bounds
\begin{align}
H_{\mathsf{X}}(A_i) & \leq H_{\mathsf{X}}(A_i\cap(A_j\cup A_k))+
H_{\mathsf{X}}(A_i\setminus(A_j\cup A_k))\label{pf3:295}\\
\mbox{and} \; H_{\mathsf{X}}(A_j) & \leq H_{\mathsf{X}}(A_j\cap(A_i\cup A_k))+
H_{\mathsf{X}}(A_j\setminus(A_i\cup A_k));\label{pf3:296}
\end{align}
\eqref{pf3:280} follows from \eqref{eq:linkCap} for $A'=A_i$, $A_j$, $A_k$, $A_i\setminus(A_j\cup A_k)$, and $A_j\setminus(A_i\cup A_k)$; and \eqref{pf3:290} follows from the fact that the capacity function $C(\cdot)$ is a modular function. Substituting \eqref{pf3:116},
\begin{align}
H_{\mathsf{W}}(I_i\cup I_j) &= nR(I_i\cup I_j) \label{pf3:297},\\
\mbox{and} \quad H_{\mathsf{W}}(I_i\cap I_j \cap I_k) &= nR(I_i\cap I_j \cap I_k) \label{pf3:298}
\end{align}
into \eqref{pf3:290} and dividing both sides of the inequality by $n$ complete the proof of \eqref{eq:GCSB3c}. \hfill $\square$

\noindent \emph{Proof of \eqref{eq:GCSB3d}.} Adding \eqref{pf3:50} and \eqref{pf3:180}, we have
\begin{align}
2H_{\mathsf{W}}(I_i\cup I_j \cup I_k) & \leq 2H_{\mathsf{X}}(A_i)+2H_{\mathsf{X}}(A_j)+2H_{\mathsf{X}}(A_k)-H_{\mathsf{X},\mathsf{W}}(A_i\cap A_j,I_i\cap I_j)-\notag\\
& \hspace{15pt} H_{\mathsf{X},\mathsf{W}}(A_i\cap A_k,I_i\cap I_k)-H_{\mathsf{X},\mathsf{W}}(A_j\cap(A_i\cup A_k),I_j\cap(I_i\cup I_k))-\notag\\
& \hspace{15pt} H_{\mathsf{X},\mathsf{W}}(A_k\cap(A_i\cup A_j),I_k\cap(I_i\cup I_j)).\label{pf3:300}
\end{align}
Note that we trivially have
\begin{align}
H_{\mathsf{X},\mathsf{W}}(A_k\cap(A_i\cup A_j),I_k\cap(I_i\cup I_j)) \geq H_{\mathsf{X}}(A_k\cap(A_i\cup A_j)).\label{pf3:301}
\end{align}
Substituting \eqref{pf3:260}, \eqref{pf3:265}, and \eqref{pf3:301} into \eqref{pf3:300}, we have
\begin{align}
2H_{\mathsf{W}}&(I_i\cup I_j \cup I_k)+H_{\mathsf{W}}(I_i\cap I_j \cap I_k)\notag\\
& \leq 2H_{\mathsf{X}}(A_i)+2H_{\mathsf{X}}(A_j)+2H_{\mathsf{X}}(A_k)-H_{\mathsf{X}}(A_i\cap(A_j\cup A_k))-\notag\\
& \hspace{15pt} H_{\mathsf{X}}(A_j\cap(A_i\cup A_k))-H_{\mathsf{X}}(A_k\cap(A_i\cup A_j))\\
& \leq H_{\mathsf{X}}(A_i)+H_{\mathsf{X}}(A_j)+H_{\mathsf{X}}(A_k)+H_{\mathsf{X}}(A_i\setminus(A_j\cup A_k))+\notag\\
& \hspace{15pt} H_{\mathsf{X}}(A_j\setminus(A_i\cup A_k))+H_{\mathsf{X}}(A_k\setminus(A_i\cup A_j))\label{pf3:310}\\
& \leq n\left(C(A_i)+C(A_j)+C(A_k)+C(A_i\setminus(A_j\cup A_k))+\right.\notag\\
& \hspace{15pt} \left.C(A_j\setminus(A_i\cup A_k))+C(A_k\setminus(A_i\cup A_j))\right)\label{pf3:320}\\
& = n\left(2C(A_i\cup A_j\cup A_k)+C(A_i\cap A_j\cap A_k)\right)\label{pf3:330}
\end{align}
where \eqref{pf3:310} follows from the independence bounds \eqref{pf3:115}, \eqref{pf3:295}, and \eqref{pf3:296}; \eqref{pf3:320} follows from \eqref{eq:linkCap} for $A'=A_i$, $A_j$, $A_k$, $A_i\setminus(A_j\cup A_k)$, $A_j\setminus(A_i\cup A_k)$ and $A_k\setminus(A_i\cup A_j)$; and \eqref{pf3:330} follows from the fact that the capacity function $C(\cdot)$ is a modular function. Substituting \eqref{pf3:116} and \eqref{pf3:298} into \eqref{pf3:330} and dividing both sides of the inequality by $n$ complete the proof of \eqref{eq:GCSB3d}. \hfill $\square$

We have thus completed the proof of Theorem~\ref{thm:GCSB3}.

\section{Generalized Cut-Set Bounds Relating $K$ Basic Cuts of the Network}\label{sec:GCSBK}
\subsection{Main Results}
\begin{theorem}\label{thm:GCSBK}
Consider a broadcast network with a collection of independent messages $\mathsf{W}_I$ collocated at the source node $s$ and $K \geq 3$ sink nodes $t_k$, $k=1,\ldots,K$. For any $k=1,\ldots,K$, let $\mathsf{W}_{I_k}$ be the intended messages for the sink node $t_k$, and let $A_k$ be a basic cut that separates the source node $s$ from the sink node $t_k$. Let $G$, $U$ and $T$ be nonempty subsets of $[K]$ such that 
\begin{align}
A^{(1)}(G) \supseteq A^{(1)}(U).
\label{eq:covCond1}
\end{align}
Let $Q$ be a subset of $\{2,\ldots,|U|\}$, and let $(r_q:q \in Q)$ be a sequence of integers from $[|T|]$ and such that 
\begin{align}
A^{(q)}(U) \subseteq A^{(r_q)}(T) \quad \mbox{and} \quad
I^{(q)}(U) \subseteq I^{(r_q)}(T), \quad \forall q \in Q.
\label{eq:covCond2}
\end{align}
We have
\begin{align}
R(I^{(1)}(G))+&\sum_{r \in \{2,\ldots,|U|\}\setminus Q}R(I^{(r)}(U))+\sum_{q \in Q}\sum_{r=1}^{r_q}\alpha_Q(q,r)R(I^{(r)}(T))\notag\\
& \leq C(A^{(1)}(G))+\sum_{r \in \{2,\ldots,|U|\}\setminus Q}C(A^{(r)}(U))+\sum_{q \in Q}\sum_{r=1}^{r_q}\alpha_Q(q,r)C(A^{(r)}(T))
\label{eq:GCSBK}
\end{align}
for any achievable rate tuple $R_I$, where 
\begin{align}
\alpha_Q(q,r) =
\left\{
\begin{array}{rl}
0, & \mbox{if} \; r\in Q\\
\frac{\prod_{\{p\in Q: p<r\}}(p-1)\prod_{\{p\in Q: r<p\leq r_q\}}p}{r_q\prod_{\{p\in Q: p\leq r_q\}}(p-1)}, & \mbox{if} \; r\notin Q
\end{array}
\right.
\label{eq:alpha}
\end{align}
for any $q \in Q$ and $r \in [r_q]$.
\end{theorem}

Note that the generalized cut-set bound \eqref{eq:GCSBK} involves a number of parameters: $G$, $U$, $T$, $Q$, and $(r_q:q \in Q)$. Specifying these parameters to certain choices will lead to potentially weaker but more applicable generalized cut-set bounds. More specifically, let $G=U=T$ and $r_q=q-1$ for any $q \in Q$. By the ordering in \eqref{eq:myOrd1}, the condition in \eqref{eq:covCond2} is satisfied (the condition in \eqref{eq:covCond1} holds trivially with an equality). Thus, by Theorem~\ref{thm:GCSBK} we have
\begin{align}
\sum_{r \in [|U|]\setminus Q}R(I^{(r)}(U))+&\sum_{q \in Q}\sum_{r=1}^{q-1}\alpha_Q(q,r)R(I^{(r)}(U))\notag\\
& \leq \sum_{r \in [|U|]\setminus Q}C(A^{(r)}(U))+\sum_{q \in Q}\sum_{r=1}^{q-1}\alpha_Q(q,r)C(A^{(r)}(U))
\label{eq:GCSBK2}
\end{align}
for any achievable rate tuple $R_I$, where
\begin{align}
\alpha_Q(q,r) =
\left\{
\begin{array}{rl}
0, & \mbox{if} \; r\in Q\\
\frac{\prod_{\{p\in Q: p<r\}}(p-1)\prod_{\{p\in Q: r<p\leq q-1\}}p}{\prod_{\{p\in Q: p\leq q\}}(p-1)}, & \mbox{if} \; r\notin Q
\end{array}
\right.
\label{eq:alpha2}
\end{align}
for any $q \in Q$ and $r \in [q-1]$. A proper simplification of \eqref{eq:GCSBK2} leads to the following corollary. See Appendix~\ref{app:pf-cor4} for the details of the simplification procedure.

\begin{coro}\label{cor:GCSBK2}
Consider a broadcast network with a collection of independent messages $\mathsf{W}_I$ collocated at the source node $s$ and $K \geq 3$ sink nodes $t_k$, $k=1,\ldots,K$. For any $k=1,\ldots,K$, let $\mathsf{W}_{I_k}$ be the intended messages for the sink node $t_k$, and let $A_k$ be a basic cut that separates the source node $s$ from the sink node $t_k$. Let $U$ be a nonempty subset of $[K]$, and let $Q$ be a subset of $\{2,\ldots,|U|\}$. We have
\begin{align}
\sum_{r=1}^{|U|}\beta_Q(r)R(I^{(r)}(U)) &\le \sum_{r=1}^{|U|}\beta_Q(r)C(A^{(r)}(U))
\label{eq:GCSBK3}
\end{align}
for any achievable rate tuple $R_I$, where $\beta_Q(r)=1$ for any $r\in[|U|]$ if $Q=\emptyset$, and
\begin{align}
\beta_Q(r)=\left\{
\begin{array}{rl}
0, & \mbox{if} \; r\in Q\\
\prod_{\{q\in Q:q<r\}}(q-1)\prod_{\{q\in Q:q>r\}}q, & \mbox{if} \; r\notin Q
\end{array}
\right.
\label{eq:beta}
\end{align}
for any $r\in[|U|]$ if $Q \neq \emptyset$.
\end{coro}

The generalized cut-set bound \eqref{eq:GCSBK3} can be further specified by letting $Q=\{2,\ldots,m\}$ for $m=1,\ldots,|U|$ (note that $Q=\emptyset$ when $m=1$). For this particular choice of $Q$, we have
\begin{align}
\beta_Q(r) &=\left\{
\begin{array}{rl}
m!, &r=1\\
0, & r=2,\ldots,m\\
(m-1)!, &r=m+1,\ldots,|U|.
\end{array}
\right.
\label{eq:beta2}
\end{align}
Substituting \eqref{eq:beta2} into \eqref{eq:GCSBK3} immediately leads to the following corollary.

\begin{coro}\label{cor:GCSBK3}
Consider a broadcast network with a collection of independent messages $\mathsf{W}_I$ collocated at the source node $s$ and $K \geq 3$ sink nodes $t_k$, $k=1,\ldots,K$. For any $k=1,\ldots,K$, let $\mathsf{W}_{I_k}$ be the intended messages for the sink node $t_k$, and let $A_k$ be a basic cut that separates the source node $s$ from the sink node $t_k$. Let $U$ be a nonempty subset of $[K]$. We have
\begin{align}
mR(I^{(1)}(U))+\sum_{r=m+1}^{|U|}R(I^{(r)}(U)) &\le mC(A^{(1)}(U))+\sum_{r=m+1}^{|U|}C(A^{(r)}(U))
\label{eq:GCSBK4}
\end{align}
for any achievable rate tuple $R_I$ and any $m=1,\ldots,|U|$.
\end{coro}

Now, the generalized cut-set bound \eqref{eq:GCSB3d} can be recovered from Corollary~\ref{cor:GCSBK3} by setting $U=\{1,2,3\}$ and $m=2$ in \eqref{eq:GCSBK4}; the generalized cut-set bound \eqref{eq:GCSB3b} can be recovered from Corollary~\ref{cor:GCSBK2} by setting $U=\{1,2,3\}$ and $Q=\{3\}$ such that
\begin{align}
\beta_Q(r)&=\left\{
\begin{array}{rl}
3, &r=1,2\\
0, &r=3;
\end{array}
\right.
\end{align}
the generalized cut-set bound \eqref{eq:GCSB3a} can be recovered from Theorem~\ref{thm:GCSBK} by setting $G=\{i,j,k\}$, $U=\{i,j\}$ (so $A^{(1)}(G) \supseteq A^{(1)}(U)$) and $Q=\emptyset$; and finally, the generalized cut-set bound \eqref{eq:GCSB3c} can be recovered from Theorem~\ref{thm:GCSBK} by setting $G=U=\{i,j,k\}$ (so $A^{(1)}(G)=A^{(1)}(U)$), $T=\{i,j\}$, $Q=\{2\}$, and $r_2=1$ such that
\begin{equation}
\begin{array}{rl}
A^{(2)}(U) & = (A_i\cap A_j)\cup(A_i\cap A_k)\cup(A_j\cap A_k) \subseteq A_i\cup A_j=A^{(r_2)}(T),\\
I^{(2)}(U) & = (I_i\cap I_j)\cup(I_i\cap I_k)\cup(I_j\cap I_k) \subseteq I_i\cup I_j=I^{(r_2)}(U),
\end{array}
\end{equation}
and $\alpha_{Q}(2,1)=1$.

\subsection{Proof of Theorem~\ref{thm:GCSBK}}\label{sec:pf-GCSBK}
Let $(n,\{\mathsf{X}_a: a \in A\})$ be an \emph{admissible} code with block length $n$, where $\mathsf{X}_a$ is the message transmitted over the arc $a$. Similar to the proof of Theorem~\ref{thm:GCSB3}, we shall assume perfect recovery of the messages at each of the sink nodes. As such, for any nonempty subset $U \subseteq [K]$ the messages $\mathsf{W}_{\cup_{k \in U}I_k}$ must be \emph{functions} of the messages $\mathsf{X}_{\cup_{k \in U}A_k}$ transmitted over the $s$-$t_U$ cut $\cup_{k \in U}A_k$.

Let us first consider the case where $Q=\emptyset$. Note that
\begin{align}
H_{\mathsf{W}}&(I^{(1)}(G))\notag\\
& \leq H_{\mathsf{X}}(A^{(1)}(G))\label{eq:pf-GCSBK100}\\
& \leq H_{\mathsf{X}}(A^{(1)}(U))+H_{\mathsf{X}}(A^{(1)}(G)\setminus A^{(1)}(U))\label{eq:pf-GCSBK200}\\
& = H_{\mathsf{X},\mathsf{W}}(A^{(1)}(U),I^{(1)}(U))+H_{\mathsf{X}}(A^{(1)}(G)\setminus A^{(1)}(U))\label{eq:pf-GCSBK300}\\
&\leq \sum_{k\in U}H_{\mathsf{X},\mathsf{W}}(A_k,I_k)-\sum_{r=2}^{|U|}H_{\mathsf{X},\mathsf{W}}(A^{(r)}(U),I^{(r)}(U))+H_{\mathsf{X}}(A^{(1)}(G)\setminus A^{(1)}(U))\label{eq:pf-GCSBK400}\\
&= \sum_{k\in U}H_{\mathsf{X}}(A_k)-\sum_{r=2}^{|U|}H_{\mathsf{X},\mathsf{W}}(A^{(r)}(U),I^{(r)}(U))+H_{\mathsf{X}}(A^{(1)}(G)\setminus A^{(1)}(U))\label{eq:pf-GCSBK500}\\
&\leq n\left(\sum_{k\in U}C(A_k)+C(A^{(1)}(G)\setminus A^{(1)}(U))\right)-\sum_{r=2}^{|U|}H_{\mathsf{X},\mathsf{W}}(A^{(r)}(U),I^{(r)}(U))\label{eq:pf-GCSBK600}\\
&= n\left(\sum_{r=1}^{|U|}C(A^{(r)}(U))+C(A^{(1)}(G)\setminus A^{(1)}(U))\right)-\sum_{r=2}^{|U|}H_{\mathsf{W}}(I^{(r)}(U))\label{eq:pf-GCSBK700}\\
&= n\left(C(A^{(1)}(G))+\sum_{r=2}^{|U|}C(A^{(r)}(U))\right)-\sum_{r=2}^{|U|}H_{\mathsf{X},\mathsf{W}}(A^{(r)}(U),I^{(r)}(U))\label{eq:pf-GCSBK800}
 \end{align}
where \eqref{eq:pf-GCSBK100} and \eqref{eq:pf-GCSBK300} follow from the fact that the messages $\mathsf{W}_{I^{(1)}(U)}$ are functions of $\mathsf{X}_{A^{(1)}(U)}$; \eqref{eq:pf-GCSBK200} follows from the independence bound on entropy; \eqref{eq:pf-GCSBK400} follows from the standard multiway submodularity \eqref{eq:submodK}; \eqref{eq:pf-GCSBK500} follows from the fact that the messages $\mathsf{W}_{I_k}$ are functions of $\mathsf{X}_{A_k}$ so we have $H_{\mathsf{X},\mathsf{W}}(A_k,I_k)=H_{\mathsf{X}}(A_k)$ for any $k \in U$; \eqref{eq:pf-GCSBK600} follows from the link capacity constraints; \eqref{eq:pf-GCSBK700} follows from the fact that the capacity function $C(\cdot)$ is a modular function so we have $\sum_{k\in U}C(A_k)=\sum_{r=1}^{|U|}C(A^{(r)}(U))$; and \eqref{eq:pf-GCSBK800} follows from the fact that the capacity function $C(\cdot)$ is a modular function and the assumption \eqref{eq:covCond1} so we have $C(A^{(1)}(G))=C(A^{(1)}(U))+C(A^{(1)}(G)\setminus A^{(1)}(U))$. Rearranging the terms in \eqref{eq:pf-GCSBK800} gives
\begin{align}
H_{\mathsf{W}}(I^{(1)}(G))+\sum_{r=2}^{|U|}H_{\mathsf{X},\mathsf{W}}(A^{(r)}(U),I^{(r)}(U)) & \leq n\left(C(A^{(1)}(G))+\sum_{r=2}^{|U|}C(A^{(r)}(U))\right).\label{eq:pf-GCSBK805}
\end{align}
Further note that
\begin{align}
H_{\mathsf{W}}(I^{(1)}(G)) &= nR(I^{(1)}(G))\label{eq:pf-GCSBK900}\\
\mbox{and} \quad H_{\mathsf{X},\mathsf{W}}(A^{(r)}(U),I^{(r)}(U)) & \geq H_{\mathsf{W}}(I^{(r)}(U)) = nR(I^{(r)}(U)), \quad \forall r=2,\ldots,|U|.\label{eq:pf-GCSBK1000} 
\end{align}
Substituting \eqref{eq:pf-GCSBK900} and \eqref{eq:pf-GCSBK1000} into \eqref{eq:pf-GCSBK800} and dividing both sides of the inequality by $n$, we have
\begin{align}
R(I^{(1)}(G))+\sum_{r=2}^{|U|}R(I^{(r)}(U)) & \leq C(A^{(1)}(G))+\sum_{r=2}^{|U|}C(A^{(r)}(U))
\end{align}
for any achievable rate tuple $R_I$. This completes the proof of \eqref{eq:GCSBK} for $Q=\emptyset$.

Next, assume that $Q \neq \emptyset$. Write, without loss of generality, that $Q=\{q_1,\ldots,q_{|Q|}\}$ where
\begin{align}
2 \leq q_1 < q_2 < \cdots < q_{|Q|} \leq |U|.
\end{align}
By Lemma~\ref{lemma:2}, for any two integers $q'$ and $r_{q'}$ such that $1 \leq q' \leq |U|$, $1 \leq r_{q'} \leq |T|$, $A^{(q')}(U) \subseteq A^{(r_{q'})}(T)$, and $I^{(q')}(U) \subseteq I^{(r_{q'})}(T)$ we have
\begin{align}
\sum_{r=1}^{r_{q'}}&H_{\mathsf{X},\mathsf{W}}(A^{(r)}(T),I^{(r)}(T))-r_{q'}H_{\mathsf{X},\mathsf{W}}(A^{(q')}(U),I^{(q')}(U))\notag\\
& \leq \sum_{r=1}^{|T|}H_{\mathsf{X},\mathsf{W}}(A_{t_r},I_{t_r})-\sum_{r=1}^{r_{q'}}H_{\mathsf{X},\mathsf{W}}(A_{t_r}\cap A^{(q')}(U),I_{t_r}\cap I^{(q')}(U))-\notag\\
& \hspace{13pt} \sum_{r=r_{q'}+1}^{|T|}H_{\mathsf{X},\mathsf{W}}(A_{t_r}\cap(A^{(q')}(U)\cup A^{(r_{q'}+1)}(\{t_1,\ldots,t_r\})),\notag\\
& \hspace{50pt} I_{t_r}\cap(I^{(q')}(U)\cup I^{(r_{q'}+1)}(\{t_1,\ldots,t_r\})))\label{eq:pf-GCSBK1100} \\
& \leq \sum_{r=1}^{|T|}H_{\mathsf{X}}(A_{t_r})- \sum_{r=1}^{r_{q'}}H_{\mathsf{X}}(A_{t_r}\cap A^{(q')}(U))-\notag\\
& \hspace{13pt} \sum_{r=r_{q'}+1}^{|T|}H_{\mathsf{X}}(A_{t_r}\cap(A^{(q')}(U)\cup A^{(r_{q'}+1)}(\{t_1,\ldots,t_r\})))\label{eq:pf-GCSBK1200} \\
& \leq \sum_{r=1}^{r_{q'}}H_{\mathsf{X}}(A_{t_r}\setminus A^{(q')}(U))+\sum_{r=r_{q'}+1}^{|T|}H_{\mathsf{X}}(A_{t_r}\setminus(A^{(q')}(U)\cup A^{(r_{q'}+1)}(\{t_1,\ldots,t_r\}))\label{eq:pf-GCSBK1300} \\
& \leq n\left(\sum_{r=1}^{r_{q'}}C(A_{t_r}\setminus A^{(q')}(U))+\sum_{r=r_{q'}+1}^{|T|}C(A_{t_r}\setminus(A^{(q')}(U)\cup A^{(r_{q'}+1)}(\{t_1,\ldots,t_r\}))\right)\label{eq:pf-GCSBK1400}\\
& = n\left(\sum_{r=1}^{|T|}C(A_{t_r})- \sum_{r=1}^{r_{q'}}C(A_{t_r}\cap A^{(q')}(U))-\right.\notag\\
& \hspace{13pt} \left.\sum_{r=r_{q'}+1}^{|T|}C(A_{t_r}\cap(A^{(q')}(U)\cup A^{(r_{q'}+1)}(\{t_1,\ldots,t_r\})))\right)\label{eq:pf-GCSBK1500} \\
& = n\left(\sum_{r=1}^{r_{q'}}C(A^{(r)}(T))-r_{q'}C(A^{(q')}(U))\right)\label{eq:pf-GCSBK1600} 
\end{align}
where \eqref{eq:pf-GCSBK1200} follows from the fact that the messages $\mathsf{W}_{I_{t_r}}$ are functions of $\mathsf{X}_{A_{t_r}}$ so we have $H_{\mathsf{X},\mathsf{W}}(A_{t_r},I_{t_r})=H_{\mathsf{X}}(A_{t_r})$ for any $r \in [|U|]$ and the trivial inequalities 
\begin{align}
H_{\mathsf{X},\mathsf{W}}&(A_{t_r}\cap A^{(q')}(U),I_{t_r}\cap I^{(q')}(U)) \geq H_{\mathsf{X}}(A_{t_r}\cap A^{(q')}(U)), \quad \forall r\in[r_{q'}]\\
\mbox{and} \quad H_{\mathsf{X},\mathsf{W}}&(A_{t_r}\cap(A^{(q')}(U)\cup A^{(r_{q'}+1)}(\{t_1,\ldots,t_r\})),I_{t_r}\cap(I^{(q')}(U)\cup I^{(r_{q'}+1)}(\{t_1,\ldots,t_r\})))\notag\\
& \geq H_{\mathsf{X}}(A_{t_r}\cap(A^{(q')}(U)\cup A^{(r_{q'}+1)}(\{t_1,\ldots,t_r\})));
\end{align}
\eqref{eq:pf-GCSBK1300} follows from the independence bound on entropy; \eqref{eq:pf-GCSBK1400} follows from the link-capacity constraints; and \eqref{eq:pf-GCSBK1500} and \eqref{eq:pf-GCSBK1600} follow from the fact that the capacity function $C(\cdot)$ is a modular function. Letting $r_{q'}=q'=q_j$ and $U=T$ in \eqref{eq:pf-GCSBK1600}, we have
\begin{align}
\sum_{r=1}^{q_j}H_{\mathsf{X},\mathsf{W}}(A^{(r)}(T),I^{(r)}(T))-&q_jH_{\mathsf{X},\mathsf{W}}(A^{(q_j)}(T),I^{(q_j)}(T))\notag\\
& \leq n\left(\sum_{r=1}^{q_j}C(A^{(r)}(T))-q_jC(A^{(q_j)}(T))\right).\label{eq:pf-GCSBK1700}
\end{align}
Let
\begin{align}
n_Q(q,r) &:= \prod_{\{p \in Q: p<r\}}(p-1) \prod_{\{p \in Q: r<p\leq r_{q}\}}p\\
\mbox{and} \quad d_Q(q) &:= \prod_{\{p \in Q: p \leq r_{q}\}}(p-1)
\end{align}
for any $q \in Q$ and $r \in [r_q]$, and let $Q_i:=\{q\in Q:q \leq r_{q_i}\}$. Note that $n_Q(q,r)$ and $d_Q(q)$ are always positive. Multiplying both sides of \eqref{eq:pf-GCSBK1700} by $n_Q(q_i,q_j)$ and then summing over all $q_j \in Q_i$, we have
\begin{align}
\sum_{j=1}^{|Q_i|}n_Q(q_i,q_j)&\left(\sum_{r=1}^{q_j}H_{\mathsf{X},\mathsf{W}}(A^{(r)}(T),I^{(r)}(T))-q_jH_{\mathsf{X},\mathsf{W}}(A^{(q_j)}(T),I^{(q_j)}(T))\right)\notag\\
\leq n&\left(\sum_{j=1}^{|Q_i|}n_Q(q_i,q_j)\left(\sum_{r=1}^{q_j}C(A^{(r)}(T))-q_jC(A^{(q_j)}(T))\right)\right).\label{eq:pf-GCSBK1800}
\end{align}
Note that
\begin{align}
\sum_{j=1}^{|Q_i|}n(q_i,q_j)\sum_{r=1}^{q_j}H_{\mathsf{X},\mathsf{W}}(A^{(r)}(T),I^{(r)}(T)) &= \sum_{r=1}^{q_{|Q_i|}}\left(\sum_{j=j(r)}^{|Q_i|}n(q_i,q_j)\right)H_{\mathsf{X},\mathsf{W}}(A^{(r)}(T),I^{(r)}(T))
\end{align}
where
\begin{align}
j(r) := \left\{
\begin{array}{rl}
1, & \mbox{for} \; 0 < r \leq q_1\\
2, & \mbox{for} \; q_1 < r \leq q_2\\
\vdots\\
|Q_i|, & \mbox{for} \; q_{|Q_i|-1} < r \leq q_{|Q_i|}.
\end{array}
\right.
\label{eq:pf-GCSBK1805}
\end{align}
We can thus rewrite \eqref{eq:pf-GCSBK1800} as
\begin{align}
\sum_{r=1}^{q_{|Q_i|}}&\left(\sum_{j=j(r)}^{|Q_i|}n(q_i,q_j)-rn_Q(q_i,r)1_{\{r \in Q_i\}}\right)H_{\mathsf{X},\mathsf{W}}(A^{(r)}(T),I^{(r)}(T))\notag\\
&\leq n\left(\sum_{r=1}^{q_{|Q_i|}}\left(\sum_{j=j(r)}^{|Q_i|}n(q_i,q_j)-rn_Q(q_i,r)1_{\{r \in Q_i\}}\right)C(A^{(r)}(T)\right).\label{eq:pf-GCSBK1900}
\end{align}
Furthermore, letting $q'=q_i$ and $r_{q'}=r_{q_i}$ in \eqref{eq:pf-GCSBK1600} and multiplying both sides of the inequality by $d_Q(q_i)$, we have
\begin{align}
\sum_{r=1}^{r_{q_i}}&d_Q(q_i)H_{\mathsf{X},\mathsf{W}}(A^{(r)}(T),I^{(r)}(T))-r_{q_i}d_Q(q_i)H_{\mathsf{X},\mathsf{W}}(A^{(q_i)}(U),I^{(q_i)}(U))\notag\\
& \leq n\left(\sum_{r=1}^{r_{q_i}}d_Q(q_i)C(A^{(r)}(T))-r_{q_i}d_Q(q_i)C(A^{(q_i)}(U))\right).\label{eq:pf-GCSBK2000}
\end{align}
Adding \eqref{eq:pf-GCSBK1900} and \eqref{eq:pf-GCSBK2000} gives
\begin{align}
\sum_{r=1}^{r_{q_i}}&n'_Q(q_i,r)H_{\mathsf{X},\mathsf{W}}(A^{(r)}(T),I^{(r)}(T))-r_{q_i}d_Q(q_i)H_{\mathsf{X},\mathsf{W}}(A^{(q_i)}(U),I^{(q_i)}(U))\notag\\
& \leq n\left(\sum_{r=1}^{r_{q_i}}n'_Q(q_i,r)C(A^{(r)}(T))-r_{q_i}d_Q(q_i)C(A^{(q_i)}(U))\right)\label{eq:pf-GCSBK2010}
\end{align}
where
\begin{align}
n'_Q(q_i,r)=\left\{
\begin{array}{rl}
\sum_{j=j(r)}^{|Q_i|}n_Q(q_i,q_j)-rn_Q(q_i,r)1_{\{r \in Q_i\}}+d_Q(q_i),  & \mbox{if} \; 1 \leq r \leq q_{|Q_i|}\\
d_Q(q_i), & \mbox{if} \; q_{|Q_i|} < r \leq  r_{q_i}.
\end{array}
\right.
\end{align}
By \eqref{eq:pf-GCSBK1805}, when $q_{m-1} < r \leq q_m$ for some $m=1,\ldots,|Q_i|$ ($q_0:=0$ for convenience), we have $j(r)=m$ and hence
\begin{align}
\sum_{j=j(r)}^{|Q_i|}n_Q(q_i,q_j) &= \sum_{j=m}^{|Q_i|}n_Q(q_i,q_j)\\
&= \sum_{j=m}^{|Q_i|}\left(\prod_{l=1}^{j-1}(q_l-1)\prod_{l=j+1}^{|Q_i|}q_l\right)\\
&= \sum_{j=m}^{|Q_i|}\left(\prod_{l=1}^{j-1}(q_l-1)\prod_{l=j}^{|Q_i|}q_l-\prod_{l=1}^{j}(q_l-1)\prod_{l=j+1}^{|Q_i|}q_l\right)\\
&= \sum_{j=m}^{|Q_i|}\left(\prod_{l=1}^{j-1}(q_l-1)\prod_{l=j}^{|Q_i|}q_l\right)-\sum_{j=m+1}^{|Q_i|+1}\left(\prod_{l=1}^{j-1}(q_l-1)\prod_{l=j}^{|Q_i|}q_l\right)\\
&= \prod_{l=1}^{m-1}(q_l-1)\prod_{l=m}^{|Q_i|}q_l-\prod_{l=1}^{|Q_i|}(q_l-1)\\
&= \prod_{l=1}^{m-1}(q_l-1)\prod_{l=m}^{|Q_i|}q_l-d_Q(q_i).
\end{align}
Therefore, when $r= q_m$ for some $m \in \left[|Q_i|\right]$ we have
\begin{align}
\sum_{j=m}^{|Q_i|}n_Q(q_i,q_j)-q_mn_Q(q_i,q_m)+d_Q(q_i) &= \prod_{l=1}^{m-1}(q_l-1)\prod_{l=m}^{|Q_i|}q_l-q_m\prod_{l=1}^{m-1}(q_l-1)\prod_{l=m+1}^{|Q_i|}q_l\\
&= 0\label{eq:pf-GCSBK2020};
\end{align}
when $q_{m-1} < r< q_m$ for some $m \in \left[|Q_i|\right]$ we have
\begin{align}
\sum_{j=m}^{|Q_i|}n_Q(q_i,q_j)+d_Q(q_i) &= \prod_{l=1}^{m-1}(q_l-1)\prod_{l=m}^{|Q_i|}q_l\\
&=r_{q_i}d_Q(q_i)\alpha_{Q}(q_i,r);\label{eq:pf-GCSBK2030}
\end{align}
and when $q_{|Q_i|} < r \leq  r_{q_i}$ we have $\alpha_{Q}(q_i,r)=1/r_{q_i}$ and hence
\begin{align}
d_Q(q_i) &= r_{q_i}d_Q(q_i)\alpha_{Q}(q_i,r).\label{eq:pf-GCSBK2040}
\end{align}
Combining \eqref{eq:pf-GCSBK2020}, \eqref{eq:pf-GCSBK2030}, and \eqref{eq:pf-GCSBK2040}, we conclude that
\begin{align}
n'_{Q}(q_i,r) &= r_{q_i}d_Q(q_i)\alpha_{Q}(q_i,r), \quad \forall r \in [r_{q_i}].
\end{align}
Dividing both sides of \eqref{eq:pf-GCSBK2010} by $r_{q_i}d(q_i)$ and then summing over all $q_i \in Q$, we have
\begin{align}
\sum_{q \in Q}\sum_{r=1}^{r_{q}}&\alpha_Q(q,r)H_{\mathsf{X},\mathsf{W}}(A^{(r)}(T),I^{(r)}(T))-\sum_{q \in Q}H_{\mathsf{X},\mathsf{W}}(A^{(q)}(U),I^{(q)}(U))\notag\\
&\leq n\left(\sum_{q \in Q}\sum_{r=1}^{r_{q}}\alpha_Q(q,r)C(A^{(r)}(T))- \sum_{q \in Q}C(A^{(q)}(U))\right).\label{eq:pf-GCSBK2050}
\end{align}
Adding \eqref{eq:pf-GCSBK805} and \eqref{eq:pf-GCSBK2050}, we have
\begin{align}
H_{\mathsf{W}}&(I^{(1)}(G))+\sum_{r\in\{2,\ldots,|U|\}\setminus Q}H_{\mathsf{X},\mathsf{W}}(A^{(r)}(U),I^{(r)}(U))+\sum_{q \in Q}\sum_{r=1}^{r_{q}}\alpha_Q(q,r)H_{\mathsf{X},\mathsf{W}}(A^{(r)}(T),I^{(r)}(T))\notag\\
& \leq n\left(C(A^{(1)}(G))+\sum_{r\in\{2,\ldots,|U|\}\setminus Q}C(A^{(r)}(U))+\sum_{q \in Q}\sum_{r=1}^{r_{q}}\alpha_Q(q,r)C(A^{(r)}(T))\right).\label{eq:pf-GCSBK2060}
\end{align}
Note that we trivially have
\begin{align}
H_{\mathsf{X},\mathsf{W}}(A^{(r)}(T),I^{(r)}(T)) & \geq H_{\mathsf{W}}(I^{(r)}(T)) = nR(I^{(r)}(T)), \quad \forall q \in Q \; \mbox{and} \; r \in [r_q].
\label{eq:pf-GCSBK2070}
\end{align}
Substituting \eqref{eq:pf-GCSBK900}, \eqref{eq:pf-GCSBK1000}, and \eqref{eq:pf-GCSBK2070} into \eqref{eq:pf-GCSBK2060} and dividing both sides of the inequality by $n$ complete the proof of \eqref{eq:GCSBK} for $Q \neq \emptyset$.

We have thus completed the proof of Theorem~\ref{thm:GCSBK}.

\section{Applications to Combination Networks}\label{sec:CN}

\begin{figure}[!t]
\centering
\includegraphics[width=0.75\linewidth,draft=false]{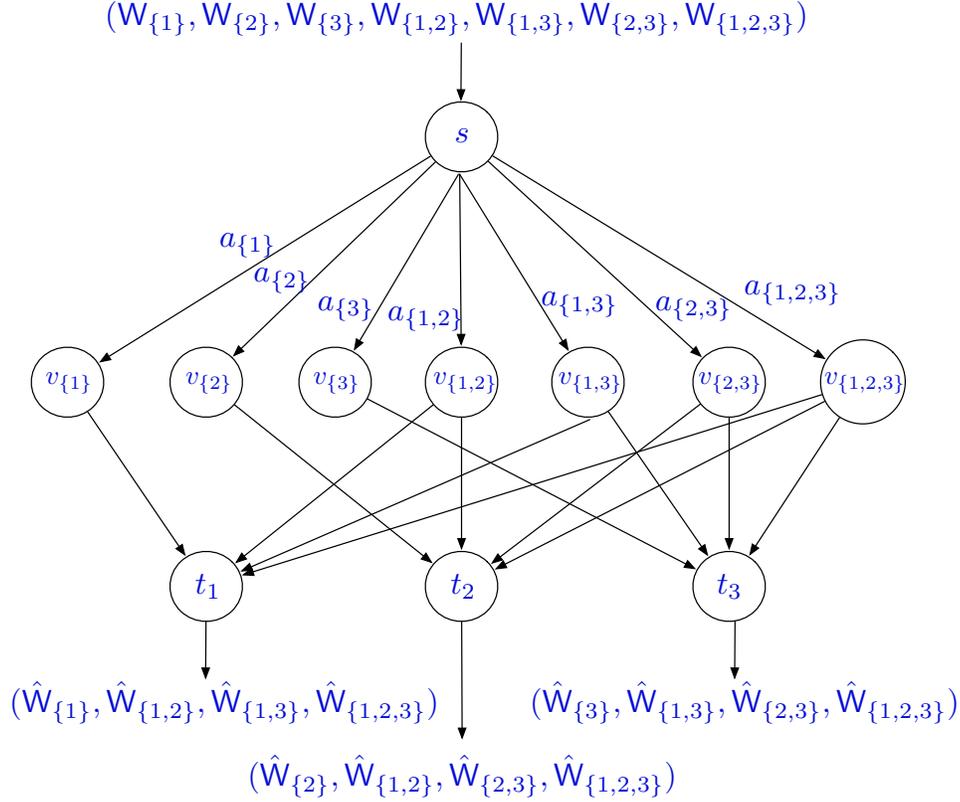}
\caption{Illustration of the general combination network with $K=3$ sink nodes and a complete message set.}
\label{fig:CN3}
\end{figure}

To demonstrate the tightness of the generalized cut-set bounds, let us consider a special class of broadcast networks known as \emph{combination} networks \cite{Nga-ITW04}. A combination network is a broadcast network that consists of three layers of nodes (see Figure~\ref{fig:CN3} for an illustration). The top layer consists of a single source node $s$, and the bottom layer consists of $K$ sink nodes $t_k$, $k=1,\ldots,K$. The middle layer consists of $2^K-1$ intermediate nodes, each connecting to the source node $s$ and a nonempty subset of sink nodes. While the links from the source node $s$ to the intermediate nodes may have finite capacity, the links from the intermediate nodes to the sink nodes are all assumed to have \emph{infinite} capacity. More specifically, denote by $v_U$ the intermediate node that connects to the nonempty subset $U$ of sink nodes and $a_U$ the link that connects the source node $s$ to the intermediate node $v_U$. The link capacity for $a_U$ is denoted by $C_U$. Note that when $C_U=0$, the intermediate node $v_U$ can be effectively removed from the network. By construction, the only interesting combinatorial structure for combination networks is cut. Therefore, combination networks provide an ideal set of problems to understand the strength and the limitations of the generalized cut-set bounds. 

In Figure~\ref{fig:CN3} we illustrate a general combination network with $K=3$ sink nodes and a general message set that consists of a total of seven independent messages 
$$(\mathsf{W}_{\{1\}},\mathsf{W}_{\{2\}},\mathsf{W}_{\{3\}},\mathsf{W}_{\{1,2\}},\mathsf{W}_{\{1,3\}},\mathsf{W}_{\{2,3\}},\mathsf{W}_{\{1,2,3\}}),$$
where the message $\mathsf{W}_U$, $U \subseteq \{1,2,3\}$, is intended for all sink nodes $t_k$, $k \in U$. This network coding problem was first introduced and solved by Grokop and Tse \cite{Gro-AP08} in the context of characterizing the \emph{latency} capacity region \cite{Tia-IT11} of the general broadcast channel with three receivers. More specifically, it was shown in \cite{Gro-AP08} that the capacity region of the network is given by the set of nonnegative rate tuples 
$$(R_{\{1\}},R_{\{2\}},R_{\{3\}},R_{\{1,2\}},R_{\{2,3\}},R_{\{1,3\}},R_{\{1,2,3\}})$$ 
satisfying 
\begin{align}
R_{\{1\}}+&R_{\{1,2\}}+R_{\{1,3\}}+R_{\{1,2,3\}} \leq C_{\{1\}}+C_{\{1,2\}}+C_{\{1,3\}}+C_{\{1,2,3\}},\label{eq:CN3-1}\\
R_{\{2\}}+&R_{\{1,2\}}+R_{\{2,3\}}+R_{\{1,2,3\}} \leq C_{\{2\}}+C_{\{1,2\}}+C_{\{2,3\}}+C_{\{1,2,3\}},\label{eq:CN3-2}\\
R_{\{3\}}+&R_{\{1,3\}}+R_{\{2,3\}}+R_{\{1,2,3\}} \leq C_{\{3\}}+C_{\{1,3\}}+C_{\{2,3\}}+C_{\{1,2,3\}},\label{eq:CN3-3}\\
R_{\{1\}}+R_{\{2\}}+&R_{\{1,2\}}+R_{\{2,3\}}+R_{\{1,3\}}+R_{\{1,2,3\}}\notag\\
& \leq C_{\{1\}}+C_{\{2\}}+C_{\{1,2\}}+C_{\{2,3\}}+C_{\{1,3\}}+C_{\{1,2,3\}},\label{eq:CN3-4}\\
R_{\{2\}}+R_{\{3\}}+&R_{\{1,2\}}+R_{\{2,3\}}+R_{\{1,3\}}+R_{\{1,2,3\}}\notag\\
& \leq C_{\{2\}}+C_{\{3\}}+C_{\{1,2\}}+C_{\{2,3\}}+C_{\{1,3\}}+C_{\{1,2,3\}},\label{eq:CN3-5}\\
R_{\{1\}}+R_{\{3\}}+&R_{\{1,2\}}+R_{\{2,3\}}+R_{\{1,3\}}+R_{\{1,2,3\}}\notag\\
& \leq C_{\{1\}}+C_{\{3\}}+C_{\{1,2\}}+C_{\{2,3\}}+C_{\{1,3\}}+C_{\{1,2,3\}},\label{eq:CN3-6}\\
R_{\{1\}}+R_{\{2\}}+R_{\{3\}}+&R_{\{1,2\}}+R_{\{2,3\}}+R_{\{1,3\}}+R_{\{1,2,3\}}\notag\\
& \leq C_{\{1\}}+C_{\{2\}}+C_{\{3\}}+C_{\{1,2\}}+C_{\{2,3\}}+C_{\{1,3\}}+C_{\{1,2,3\}},\label{eq:CN3-7}\\
R_{\{1\}}+R_{\{2\}}+R_{\{3\}}+&2R_{\{1,2\}}+R_{\{2,3\}}+R_{\{1,3\}}+2R_{\{1,2,3\}}\notag\\
& \leq C_{\{1\}}+C_{\{2\}}+C_{\{3\}}+2C_{\{1,2\}}+C_{\{2,3\}}+C_{\{1,3\}}+2C_{\{1,2,3\}},\label{eq:CN3-8}\\
R_{\{1\}}+R_{\{2\}}+R_{\{3\}}+&R_{\{1,2\}}+2R_{\{2,3\}}+R_{\{1,3\}}+2R_{\{1,2,3\}}\notag\\
& \leq C_{\{1\}}+C_{\{2\}}+C_{\{3\}}+C_{\{1,2\}}+2C_{\{2,3\}}+C_{\{1,3\}}+2C_{\{1,2,3\}},\label{eq:CN3-9}\\
R_{\{1\}}+R_{\{2\}}+R_{\{3\}}+&R_{\{1,2\}}+R_{\{2,3\}}+2R_{\{1,3\}}+2R_{\{1,2,3\}}\notag\\
& \leq C_{\{1\}}+C_{\{2\}}+C_{\{3\}}+C_{\{1,2\}}+C_{\{2,3\}}+2C_{\{1,3\}}+2C_{\{1,2,3\}},\label{eq:CN3-10}\\
R_{\{1\}}+R_{\{2\}}+R_{\{3\}}+&2R_{\{1,2\}}+2R_{\{2,3\}}+2R_{\{1,3\}}+2R_{\{1,2,3\}}\notag\\
& \leq C_{\{1\}}+C_{\{2\}}+C_{\{3\}}+2C_{\{1,2\}}+2C_{\{2,3\}}+2C_{\{1,3\}}+2C_{\{1,2,3\}},\label{eq:CN3-11}\\
R_{\{1\}}+2R_{\{2\}}+2R_{\{3\}}+&2R_{\{1,2\}}+2R_{\{2,3\}}+2R_{\{1,3\}}+3R_{\{1,2,3\}}\notag\\
& \leq C_{\{1\}}+2C_{\{2\}}+2C_{\{3\}}+2C_{\{1,2\}}+2C_{\{2,3\}}+2C_{\{1,3\}}+3C_{\{1,2,3\}},\label{eq:CN3-12}\\
2R_{\{1\}}+R_{\{2\}}+2R_{\{3\}}+&2R_{\{1,2\}}+2R_{\{2,3\}}+2R_{\{1,3\}}+3R_{\{1,2,3\}}\notag\\
& \leq 2C_{\{1\}}+2C_{\{2\}}+C_{\{3\}}+2C_{\{1,2\}}+2C_{\{2,3\}}+2C_{\{1,3\}}+3C_{\{1,2,3\}},\label{eq:CN3-13}\\
2R_{\{1\}}+2R_{\{2\}}+R_{\{3\}}+&2R_{\{1,2\}}+2R_{\{2,3\}}+2R_{\{1,3\}}+3R_{\{1,2,3\}}\notag\\
& \leq 2C_{\{1\}}+2C_{\{2\}}+C_{\{3\}}+2C_{\{1,2\}}+2C_{\{2,3\}}+2C_{\{1,3\}}+3C_{\{1,2,3\}},\label{eq:CN3-14}\\
2R_{\{1\}}+2R_{\{2\}}+2R_{\{3\}}+&2R_{\{1,2\}}+2R_{\{2,3\}}+2R_{\{1,3\}}+3R_{\{1,2,3\}}\notag\\
& \leq 2C_{\{1\}}+2C_{\{2\}}+2C_{\{3\}}+2C_{\{1,2\}}+2C_{\{2,3\}}+2C_{\{1,3\}}+3C_{\{1,2,3\}}.\label{eq:CN3-15} 
\end{align}
From the converse viewpoint, the inequalities \eqref{eq:CN3-1}--\eqref{eq:CN3-7} follow directly from the standard cut-set bounds \eqref{eq:CSB} by considering the following three basic cuts: $A_1=\{a_{\{1\}},a_{\{1,2\}},a_{\{1,3\}},a_{\{1,2,3\}}\}$, $A_2=\{a_{\{2\}},a_{\{1,2\}},a_{\{2,3\}},a_{\{1,2,3\}}\}$, and $A_3=\{a_{\{3\}},a_{\{2,3\}},a_{\{1,3\}},a_{\{1,2,3\}}\}$. For the inequalities \eqref{eq:CN3-8}--\eqref{eq:CN3-15}, the proof provided in \cite{Gro-AP08} was problem-specific and appears to be rather hand-crafted. With the generalized cut-set bounds now in place, however, it is clear that the inequalities \eqref{eq:CN3-8}--\eqref{eq:CN3-10} follow directly from \eqref{eq:GCSB3a}; the inequality \eqref{eq:CN3-11} follows directly from \eqref{eq:GCSB3b}; the inequalities \eqref{eq:CN3-12}--\eqref{eq:CN3-14} follow directly from \eqref{eq:GCSB3c}; and the inequality \eqref{eq:CN3-15} follows directly from \eqref{eq:GCSB3d}. Thus, the standard and the generalized cut-set bounds together provide an \emph{exact} characterization of the capacity region of the general combination network with three sink nodes and a complete message set.

Next, let us consider the general combination network with $K$ sink nodes and \emph{symmetrical} link capacity constraints \cite{Tia-IT11}:
\begin{align}
C_U=C_{|U|}, \quad \forall U \subseteq [K]
\end{align}
i.e., the link-capacity constraint for arc $a_U$ depends on the subset $U$ only via its cardinality. Assume that the source $s$ has access to a set of $K+1$ independent messages $(\mathsf{W}_1,\ldots,\mathsf{W}_K,\mathsf{W}_0)$, where $\mathsf{W}_k$, $k=1,\ldots,K$, is a private message intended only for the sink node $t_k$, and $\mathsf{W}_0$ is a common message intended for all $K$ sink nodes in the network. For this communication scenario, note that $A_k=\{a_U: U\ni k\}$ is a basic cut that separates the source node $s$ from the sink node $t_k$ for each $k=1,\ldots,K$. Applying Corollary~\ref{cor:GCSBK3} with $U=[K]$, we have
\begin{align}
KR_0+mR_{sp} & \leq m\sum_{r=1}^K
\left(
\begin{array}{c}
  K   \\
  r   
\end{array}
\right)
C_r+\sum_{r=m+1}^K\sum_{j=r}^K\left(
\begin{array}{c}
  K   \\
  j   
\end{array}
\right)
C_j\\
& = 
m\sum_{r=1}^K
\left(
\begin{array}{c}
  K   \\
  r   
\end{array}
\right)
C_r+\sum_{r=m+1}^K(r-m)\left(
\begin{array}{c}
  K   \\
  r   
\end{array}
\right)
C_r
\label{eq:CNK-1}
\end{align}
for any achievable rate tuple $(R_0,R_1,\ldots,R_K)$ and any $m=1,\ldots,K$, where $R_{sp}=\sum_{k=1}^{K}R_k$ is the sum of the private rates. It is clear that the outer bound given by the inequality \eqref{eq:CNK-1} for $m=1,\ldots,K$ has exactly $K+1$ corner points:
$$\left(\sum_{i=r}^K
\left(
\begin{array}{c}
  K-1   \\
  i-1   
\end{array}
\right)
C_i,\sum_{i=1}^{r-1}
\left(
\begin{array}{c}
  K   \\
  i   
\end{array}
\right)
C_i\right), \quad r=1,\ldots,K+1.$$
The achievability of these corner points was proved in \cite{Tia-IT11}. Therefore, the generalized cut-set bounds also provide a \emph{tight} characterization of the common-v.s.-sum-private capacity region of the general symmetrical combination network.

\begin{figure}[!t]
\centering
\includegraphics[width=0.5\linewidth,draft=false]{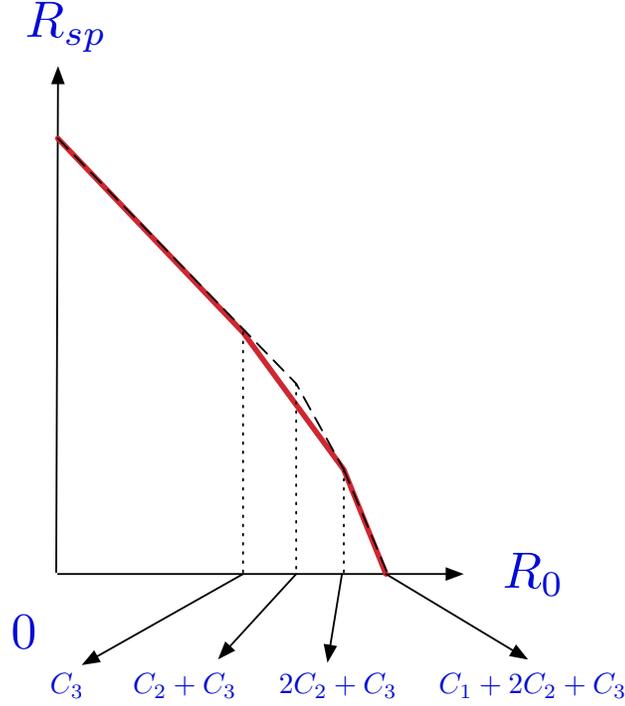}
\caption{Capacity v.s. cut-set outer regions for $K=3$ sinks. The boundary of the capacity region is illustrated by solid lines, while the boundary of the cut-set outer region is illustrated by dashed lines.}
\label{fig:CNK}
\end{figure}

Finally, let us make an explicit comparison between the common-v.s.-sum-private capacity region of the general symmetrical combination network and the outer region given by \emph{just} the standard cut-set bounds for the case of $K=3$ sink nodes. For $K=3$, the common-v.s.-sum-private capacity region of the network is given by all nonnegative $(R_0,R_{sp})$ pairs satisfying
\begin{align}
\begin{array}{rcl}
3R_0+R_{sp} & \leq & 3C_1+6C_2+3C_3,\\
3R_0+2R_{sp} & \leq & 6C_1+6C_2+3C_3,\\
\mbox{and} \quad R_0+R_{sp} & \leq & 3C_1+3C_2+C_3.
\end{array}
\label{eq:CNK-2}
\end{align}
The standard cut-set bounds, in this case, are given by
\begin{align}
\begin{array}{rcl}
R_0+R_1 & \leq & C_1+2C_2+C_3,\\
R_0+R_2 & \leq & C_1+2C_2+C_3,\\
R_0+R_3 & \leq & C_1+2C_2+C_3,\\
R_0+R_1+R_2 & \leq & 2C_1+3C_2+C_3,\\
R_0+R_1+R_3 & \leq & 2C_1+3C_2+C_3,\\
R_0+R_3+R_2 & \leq & 2C_1+3C_2+C_3,\\
R_0+R_1+R_2+R_3 & \leq & 2C_1+3C_2+C_3.
\end{array}
\label{eq:CNK-3}
\end{align}
Substituting $R_1=R_{sp}-R_2-R_3$ into \eqref{eq:CNK-3} and using Fourier-Motzkin elimination to eliminate $R_2$ and $R_3$ from the inequalities in \eqref{eq:CNK-3}, we may explicitly write the outer region given by just the standard cut-set bounds as the nonnegative $(R_0,R_{sp})$ pairs satisfying
\begin{equation}
\begin{array}{rcl}
3R_0+R_{sp} & \leq & 3C_1+6C_2+3C_3,\\
2R_0+R_{sp} & \leq & 3C_1+5C_2+2C_3,\\
\mbox{and} \quad R_0+R_{sp} & \leq & 3C_1+3C_2+C_3.
\end{array}
\label{eq:CNK-4}
\end{equation}
In Figure~\ref{fig:CNK} we illustrate the rate regions constrained by \eqref{eq:CNK-2} and \eqref{eq:CNK-4}, respectively. Clearly, even for the case with only $K=3$ sink nodes, the standard cut-set bounds alone are \emph{not} tight, while the generalized cut-set bounds provide a precise characterization of the common-v.s.-sum-private capacity region. 

\section{Concluding Remarks}\label{sec:Con}
The paper considered the problem of coding over broadcast networks with multiple (multicast) messages and more than two sink nodes. The standard cut-set bounds, which are known to be loose in general, are closely related to union as a specific set operation to combine different basic cuts of the network. A new set of network coding bounds (termed as \emph{generalized} cut-set bounds), which relate the basic cuts of the network via a variety of set operations (not just the union), were established via the submodularity of the Shannon entropy. It was shown that the generalized cut-set bounds (together with the standard cut-set bounds) provide a precise characterization of the capacity region of the general combination network with three sink nodes and the common-v.s.-sum-private capacity region of the general symmetrical combination network (with arbitrary number of sink nodes).

Our ongoing work focuses primarily on further understanding the strength and the limitations of the generalized cut-set bounds established in this paper. In particular, it would be interesting to see whether the generalized cut-set bounds are tight for the \emph{symmetrical} capacity region of the general symmetrical combination network, which was recently characterized by Tian \cite{Tia-IT11}.

\appendix

\section{Proof of Lemma~\ref{lemma:1}}\label{app:pf-lemma1}
Fix two integers $r'$ and $J$ such that $0 < r' < J \leq K$. Let
\begin{align}
T_r := \left\{
\begin{array}{rl}
\emptyset, & \mbox{for} \; r=1,\ldots,r'\\
S^{(r'+1)}([r]), & \mbox{for} \; r=r'+1,\ldots,J,
\end{array}
\right.
\end{align}
and let $G_r:=S_r \cup T_r$ for $r=1,\ldots,J$. By the standard multiway submodularity \eqref{eq:submodK} and modularity \eqref{eq:modK} we have
\begin{align}
\sum_{r=1}^{r'}f(S_r)+\sum_{r=r'+1}^{J}f(S_r\cup S^{(r'+1)}([r])) & =\sum_{r=1}^{J}f(G_r) \geq \sum_{r=1}^{J}f(G^{(r)}([J]))\label{eq:ASsubmod}
\end{align}
if $f$ is a submodular function, and
\begin{align}
\sum_{r=1}^{r'}f(S_r)+\sum_{r=r'+1}^{J}f(S_r\cup S^{(r'+1)}([r])) & =\sum_{r=1}^{J}f(G_r) = \sum_{r=1}^{J}f(G^{(r)}([J]))
\label{eq:ASmod}
\end{align}
if $f$ is a modular function. Next, we shall show that 
\begin{align}
G^{(r)}([J]) = \left\{
\begin{array}{rl}
S^{(r)}([J]), & \mbox{for} \; r=1,\ldots,r'\\
S^{(r'+1)}([J-r+r'+1]), & \mbox{for} \; r=r'+1,\ldots,J.
\end{array}
\right.
\label{eq:AS}
\end{align}

We shall consider the following two cases separately.

Case 1: $r \in [r']$. Note that $S_r \subseteq G_r$ for any $r \in [J]$, so we have $S^{(r)}([J]) \subseteq G^{(r)}([J])$ for any $r \in [J]$. On the other hand, since $T_r \subseteq S^{(r'+1)}([J])$ for all $r \in [J]$, we have $G_r \subseteq S_r \cup S^{(r'+1)}([J])$ and hence $G^{(r)}([J]) \subseteq S^{(r)}([J]) \cup S^{(r'+1)}([J])$ for all $r \in [J]$. Since $S^{(r)}([J]) \supseteq S^{(r'+1)}([J])$ for all $r \in [r']$, we have $G^{(r)}([J]) \subseteq S^{(r)}([J])$ for all $r \in [r']$. We thus conclude that $G^{(r)}([J]) = S^{(r)}([J])$ for all $r \in [r']$.

Case 2: $r \in \{r'+1,\ldots,J\}$. For this case, we have the following fact.

\begin{fact}\label{fact}
For any $r \in \{r'+1,\ldots,J\}$, we have
\begin{align}
G^{(r)}([J]) & = \cup_{m=1}^{\min\{r,r'+2\}}\left(S^{(m-1)}([J-r+m-1])\cap T_{J-r+m}\right).
\end{align}
\end{fact}

\begin{proof}
Fix $r \in \{r'+1,\ldots,J\}$. By definition,
\begin{align}
G^{(r)}([J]) = \cup_{\{U \subseteq [J]:|U|=r\}}\cap_{k \in U}G_k.
\label{eq:AS0}
\end{align}
Fix $U \subseteq [J]$ such that $|U|=r$. We have
\begin{align}
\cap_{k \in U}G_k & = \cap_{k \in U}\left(S_k\cup T_k\right)\\
& = \cup_{U' \subseteq U}\left((\cap_{k \in U'}S_k)\cap(\cap_{k \in U\setminus U'}T_k)\right)\\
& = \left(\cup_{U' \subset U}\left((\cap_{k \in U'}S_k)\cap T_{\bar{k}(U')}\right)\right)\cup\left(\cap_{k \in U}S_k\right)\label{eq:AS1}
\end{align}
where $\bar{k}(U')$ is the \emph{smallest} integer in $U\setminus U'$, and \eqref{eq:AS1} follows from the fact that
\begin{align}
T_1 \subseteq T_{2} \subseteq \cdots \subseteq T_{J}.
\label{eq:ordT}
\end{align}
Write, without loss of generality, that $U=\{u_1,\ldots,u_r\}$ where $1 \leq u_1 < u_2 < \cdots < u_r \leq J$. Fix $\bar{k}(U')=u_m$ for some $m \in [r]$. Then we must have $U' \supseteq \{u_1,\ldots,u_{m-1}\}$ for any such $U'$. We thus have from \eqref{eq:AS1} that
\begin{align}
\cap_{k \in U}G_k=\left(\cup_{m=1}^{r}\left((\cap_{l=1}^{m-1}S_{u_l})\cap T_{u_m}\right)\right)\cup\left(\cap_{l=1}^{r}S_{u_l}\right).
\label{eq:AS2}
\end{align}

The right-hand side of \eqref{eq:AS2} can be further simplified based on the following two observations. First, for any $r \in \{r'+1,\ldots,J\}$ we have $u_r \geq r \geq r'+1$ and hence
\begin{align}
T_{u_r} &=S^{(r'+1)}([u_r]) \supseteq \cap_{l=1}^{r'+1}S_{u_l} \supseteq \cap_{l=1}^{r}S_{u_l}.
\end{align}
We thus have
\begin{align}
\cap_{l=1}^{r}S_{u_l} \subseteq (\cap_{l=1}^{r-1}S_{u_l})\cap T_{u_r}
\end{align}
and hence
\begin{align}
\cap_{k \in U}G_k=\cup_{m=1}^{r}\left((\cap_{l=1}^{m-1}S_{u_l})\cap T_{u_m}\right).
\label{eq:AS3}
\end{align}
Second, since $u_{r'+2} \geq r'+2$, we have
\begin{align}
\cap_{l=1}^{r'+1}S_{u_l} \subseteq S^{(r'+1)}([u_{r'+2}])=T_{u_{r'+2}}
\end{align}
and hence
\begin{align}
(\cap_{l=1}^{r'+1}S_{u_l})\cap T_{u_{r'+2}}=\cap_{l=1}^{r'+1}S_{u_l}.
\end{align}
It follows that for any $m \geq r'+2$, we have
\begin{align}
(\cap_{l=1}^{m-1}S_{u_l})\cap T_{u_m} \subseteq \cap_{l=1}^{r'+1}S_{u_l}= (\cap_{l=1}^{r'+1}S_{u_l})\cap T_{u_{r'+2}}.
\label{eq:AS3.5}
\end{align}
Substituting \eqref{eq:AS3.5} into \eqref{eq:AS3}, we have
\begin{align}
\cap_{k \in U}G_k=\cup_{m=1}^{\min\{r,r'+2\}}\left((\cap_{l=1}^{m-1}S_{u_l})\cap T_{u_m}\right).
\label{eq:AS3.6}
\end{align}

Finally, substituting \eqref{eq:AS3.6} into \eqref{eq:AS0}, we have
\begin{align}
G^{(r)}([J]) &= \cup_{\{U \subseteq [J]:|U|=r\}}\left(\cup_{m=1}^{\min\{r,r'+2\}}\left((\cap_{l=1}^{m-1}S_{u_l})\cap T_{u_m}\right)\right)\\
&= \cup_{m=1}^{\min\{r,r'+2\}}\left(\cup_{\{U \subseteq [J]:|U|=r\}}\left((\cap_{l=1}^{m-1}S_{u_l})\cap T_{u_m}\right)\right)
\label{eq:AS4}
\end{align}
for any $r \in \{r'+1,\ldots,J\}$. Note that for any $U \subseteq [J]$ such that $|U|=r$, the largest numerical value that $u_m$ can assume is  $J-r+m$ for any $m\in [r]$. By the ordering in \eqref{eq:ordT}, for any $m=1,\ldots,r$ we have
\begin{align}
\cup_{\{U \subseteq [J]:|U|=r\}}\left((\cap_{l=1}^{m-1}S_{u_l})\cap T_{J-r+m}\right) 
&= \left(\cup_{\{1 \leq u_1 < u_2 < \cdots < u_{m-1} \leq J-r+m-1\}}\cap_{l=1}^{m-1}S_{u_l}\right)\cap T_{J-r+m}\\
&= S^{(m-1)}([J-r+m-1])\cap T_{J-r+m}.
\label{eq:AS5}
\end{align}
Substituting \eqref{eq:AS5} into \eqref{eq:AS4} completes the proof of the fact.
\end{proof}

Further note that for any $r \in \{r'+1,\ldots,J\}$ we have
\begin{align}
T_{J-r+m} \subseteq S^{(r'+1)}([J-r+m]) \subseteq S^{(r')}([J-r+m-1]) \subseteq S^{(m-1)}([J-r+m-1])
\label{eq:AS6}
\end{align}
for any $2 \leq m \leq r'+1$. When $r=r'+1$, substituting \eqref{eq:ordT} and \eqref{eq:AS6} into Fact~\ref{fact} we have
\begin{align}
G^{(r)}([J]) &= \cup_{m=1}^{r}T_{J-r+m}=T_J=S^{(r'+1)}([J]).
\label{eq:AS7}
\end{align}
When $r \in \{r'+2,\ldots,J\}$, by Fact~\ref{fact} we have
\begin{align}
G^{(r)}([J]) &= \cup_{m=1}^{r'+2}\left(S^{(m-1)}([J-r+m-1])\cap T_{J-r+m}\right)\\
&= \left(\cup_{m=1}^{r'+1}\left(S^{(m-1)}([J-r+m-1])\cap T_{J-r+m}\right)\right)\cup\notag\\
& \hspace{13pt} \left(S^{(r'+1)}([J-r+r'+1])\cap T_{J-r+r'+2}\right)\\
&= \left(\cup_{m=1}^{r'+1}T_{J-r+m}\right)\cup\left(S^{(r'+1)}([J-r+r'+1])\cap T_{J-r+r'+2}\right)\label{eq:AS7.1}\\
&= T_{J-r+r'+1}\cup\left(S^{(r'+1)}([J-r+r'+1])\cap T_{J-r+r'+2}\right)\label{eq:AS7.2}\\
&= S^{(r'+1)}([J-r+r'+1])\cup\left(S^{(r'+1)}([J-r+r'+1])\cap S^{(r'+1)}([J-r+r'+2])\right)\label{eq:AS7.3}\\
&= S^{(r'+1)}([J-r+r'+1])\label{eq:AS8}
\end{align}
where \eqref{eq:AS7.1} follows from \eqref{eq:AS6}, and \eqref{eq:AS7.2} follows from the ordering in \eqref{eq:ordT}. Combining \eqref{eq:AS7} and \eqref{eq:AS8} completes the proof of \eqref{eq:AS} for $r \in \{r'+1,\ldots,J\}$.

Finally, substituting \eqref{eq:AS} into \eqref{eq:ASsubmod} and \eqref{eq:ASmod} we have
\begin{align}
\sum_{r=1}^{r'}f(S_r)+\sum_{r=r'+1}^{J}f(S_r\cup S^{(r'+1)}([r])) & \geq \sum_{r=1}^{r'}f(S^{(r)}([J]))+\sum_{r=r'+1}^{J}f(S^{(r'+1)}([J-r+r'+1]))\\
& = \sum_{r=1}^{r'}f(S^{(r)}([J]))+\sum_{r=r'+1}^{J}f(S^{(r'+1)}([r]))
\end{align}
if $f$ is a submodular function, and
\begin{align}
\sum_{r=1}^{r'}f(S_r)+\sum_{r=r'+1}^{J}f(S_r\cup S^{(r'+1)}([r])) & =\sum_{r=1}^{r'}f(S^{(r)}([J]))+\sum_{r=r'+1}^{J}f(S^{(r'+1)}([r]))
\end{align}
if $f$ is a modular function. This completes the proof of Lemma~\ref{lemma:1}.

\section{Proof of Lemma~\ref{lemma:2}}\label{app:pf-lemma2}
Without loss of generality, we may assume that $T=[|T|]$ such that $t_r=r$ for all $r=1,\ldots,|T|$. Under this assumption, the inequality \eqref{eq:3a} can be written as
\begin{align}
\sum_{r=1}^{|T|}&f(S_r)+r_qf(S^{(q)}(U))\notag\\
& \ge \sum_{r=1}^{r_q}\left(f(S^{(r)}(T))+f(S_r\cap S^{(q)}(U))\right)+\sum_{r=r_q+1}^{|T|}f(S_r\cap(S^{(q)}(U)\cup S^{(r_q+1)}([r]))).
\label{eq:3a'}
\end{align}

Assume that $f$ is a modular function. By the two-way submodularity \eqref{eq:submod2} we have
\begin{align}
\sum_{r=1}^{|T|}&f(S_r)+r_qf(S^{(q)}(U))\notag\\
& =\sum_{r=1}^{r_q}\left(f(S_r)+f(S^{(q)}(U))\right)+
\sum_{r=r_q+1}^{|T|}\left(f(S_r)+f(S^{(q)}(U)\cup S^{(r_q+1)}([r]))\right)-\notag\\
& \hspace{12pt} \sum_{r=r_q+1}^{|T|}f(S^{(q)}(U)\cup S^{(r_q+1)}([r]))\label{eq:3c}\\ 
& \geq\sum_{r=1}^{r_q}\left(f(S_r\cap S^{(q)}(U))+f(S_r\cup S^{(q)}(U))\right)+\notag\\
& \hspace{12pt} \sum_{r=r_q+1}^{|T|}\left(f(S_r\cap(S^{(q)}(U)\cup S^{(r_q+1)}([r])))+f(S_r\cup(S^{(q)}(U)\cup S^{(r_q+1)}([r])))\right)-\notag\\
& \hspace{12pt} \sum_{r=r_q+1}^{|T|}f(S^{(q)}(U)\cup S^{(r_q+1)}([r])).\label{eq:3d}
\end{align}
Applying Corollary~\ref{cor:1} with $r'=r_q$, $J=|T|$, and $S_0=S^{(q)}(U)$, we have
\begin{align}
\sum_{r=1}^{r_q}&f(S_r\cup S^{(q)}(U))+\sum_{r=r_q+1}^{|T|}f(S_r\cup S^{(r_q+1)}([r])\cup S^{(q)}(U))\notag\\
& \ge
\sum_{r=1}^{r_q}f(S^{(r)}(T)\cup S^{(q)}(U))+\sum_{r=r_q+1}^{|T|}f(S^{(r_q+1)}([r])\cup S^{(q)}(U))\label{eq:3e}\\
& =
\sum_{r=1}^{r_q}f(S^{(r)}(T))+\sum_{r=r_q+1}^{|T|}f(S^{(r_q+1)}([r])\cup S^{(q)}(U))\label{eq:3f}
\end{align}
where \eqref{eq:3f} follows from the assumption $S^{(r_q)}(T) \supseteq S^{(q)}(U)$ such that $S^{(r)}(T) \supseteq S^{(q)}(U)$ for any $r=1,\ldots,r_q$. Substituting \eqref{eq:3f} into \eqref{eq:3d} completes the proof of \eqref{eq:3a'} and hence that of \eqref{eq:3a}. 

When $f$ is a modular function, both inequalities \eqref{eq:3d} and \eqref{eq:3e} hold with an equality. This completes the proof of \eqref{eq:3b} and hence that of the entire corollary.

\section{Proof of Corollary~\ref{cor:GCSBK2}}\label{app:pf-cor4}
Note that when $Q=\emptyset$, $\beta_Q(r)=1$ for all $r \in [|U|]$. In this case, the corollary follows directly from \eqref{eq:GCSBK2}. Now, assume that $Q$ is nonempty. Write, without loss of generality, that $Q=\{q_1,\ldots,q_{|Q|}\}$ where 
\begin{align}
1=:q_0 < q_1 < q_2 < \cdots < q_{|Q|} \leq |U|.
\end{align} 
Note that
\begin{align}
\sum_{q \in Q}\sum_{r=1}^{q-1}\alpha_Q(q,r)R(I^{(r)}(U)) = \sum_{r=1}^{q_{|Q|}-1}\beta'_Q(r)R(I^{(r)}(U))
\label{eq:T2}
\end{align}
where
\begin{align}
\beta'_Q(r) &= \sum_{l=m}^{|Q|}\alpha_Q(q_l,r)
\label{eq:beta'}
\end{align}
for any $q_{m-1} \leq r < q_m$ for some $m \in [|Q|]$. When $r = q_m$ for some $m \in [|Q|-1]$, by \eqref{eq:alpha2} and \eqref{eq:beta'} we have $\alpha_Q(q_l,r)=0$ for any $l=m,\ldots,|Q|$ and hence 
\begin{align}
\beta'_Q(r)=0.
\label{eq:beta'2}
\end{align}
When $q_{m-1} < r < q_m$ for some $m \in [|Q|]$, by \eqref{eq:alpha2} and \eqref{eq:beta'} we have
\begin{align}
\alpha_Q(q_l,r) &= \frac{\prod_{t=1}^{m-1}(q_t-1)\prod_{t=m}^{l-1}q_t}{\prod_{t=1}^{l}(q_t-1)}
\end{align}
for any $l=m,\ldots,|Q|$ and hence
\begin{align}
\beta'_Q(r) &= \sum_{l=m}^{|Q|}\frac{\prod_{t=1}^{m-1}(q_t-1)\prod_{t=m}^{l-1}q_t}{\prod_{t=1}^{l}(q_t-1)}\\
&= \frac{\prod_{t=1}^{m-1}(q_t-1)}{\prod_{t=1}^{|Q|}(q_t-1)}\sum_{l=m}^{|Q|}\left(\prod_{t=m}^{l-1}q_t\prod_{t=l+1}^{|Q|}(q_t-1)\right)\\
&= \frac{\prod_{t=1}^{m-1}(q_t-1)}{\prod_{t=1}^{|Q|}(q_t-1)}\sum_{l=m}^{|Q|}\left((q_l-(q_l-1))\prod_{t=m}^{l-1}q_t\prod_{t=l+1}^{|Q|}(q_t-1)\right)\\
&= \frac{\prod_{t=1}^{m-1}(q_t-1)}{\prod_{t=1}^{|Q|}(q_t-1)}\sum_{l=m}^{|Q|}\left(\prod_{t=m}^{l}q_t\prod_{t=l+1}^{|Q|}(q_t-1)-\prod_{t=m}^{l-1}q_t\prod_{t=l}^{|Q|}(q_t-1)\right)\\
&= \frac{\prod_{t=1}^{m-1}(q_t-1)}{\prod_{t=1}^{|Q|}(q_t-1)}\left(\prod_{t=m}^{|Q|}q_t-\prod_{t=m}^{|Q|}(q_t-1)\right)\\
&= \frac{\prod_{t=1}^{m-1}(q_t-1)\prod_{t=m}^{|Q|}q_t}{\prod_{t=1}^{|Q|}(q_t-1)}-1\\
&= \frac{\beta_Q(r)}{\prod_{t=1}^{|Q|}(q_t-1)}-1,\label{eq:beta'3}
\end{align}
where \eqref{eq:beta'3} follows from the fact that 
\begin{align}
\beta_Q(r) = \prod_{t=1}^{m-1}(q_t-1)\prod_{t=m}^{|Q|}q_t, \quad \forall q_{m-1} < r < q_m
\end{align}
by the definition \eqref{eq:beta} of $\beta_Q(r)$.

By \eqref{eq:T2}, \eqref{eq:beta'2}, and \eqref{eq:beta'3}, the left-hand side of \eqref{eq:GCSBK2} can be simplified as 
\begin{align}
\sum_{r \in [|U|]\setminus Q}&R(I^{(r)}(U))+\sum_{q \in Q}\sum_{r=1}^{q-1}\alpha_Q(q,r)R(I^{(r)}(U))\notag\\
&= \sum_{r \in [|U|]\setminus Q}R(I^{(r)}(U))+\sum_{r=1}^{q_{|Q|}-1}\beta'_Q(r)R(I^{(r)}(U))\\
&= \sum_{r \in [|U|]\setminus Q}R(I^{(r)}(U))+\sum_{r \in [q_{|Q|}]\setminus Q}\left(\frac{\beta_Q(r)}{\prod_{t=1}^{|Q|}(q_t-1)}-1\right)R(I^{(r)}(U))\label{eq:T100}\\
&= \frac{1}{\prod_{t=1}^{|Q|}(q_t-1)}\left(\sum_{r \in [q_{|Q|}]\setminus Q}\beta_Q(r)R(I^{(r)}(U))+\left(\prod_{t=1}^{|Q|}(q_t-1)\right)\sum_{r=q_{|Q|}+1}^{|U|}R(I^{(r)}(U))\right)\label{eq:T200}\\
&= \frac{1}{\prod_{t=1}^{|Q|}(q_t-1)}\left(\sum_{r=1}^{q_{|Q|}}\beta_Q(r)R(I^{(r)}(U))+\sum_{r=q_{|Q|}+1}^{|U|}\beta_Q(r)R(I^{(r)}(U))\right)\label{eq:T300}\\
&= \frac{1}{\prod_{t=1}^{|Q|}(q_t-1)}\sum_{r=1}^{|U|}\beta_Q(r)R(I^{(r)}(U)),\label{eq:T400}
\end{align}
where \eqref{eq:T300} follows from the facts that $\beta_Q(r)=0$ for all $r \in Q$ and that
\begin{align}
\beta_Q(r)=\prod_{t=1}^{|Q|}(q_t-1), \quad \forall r \geq q_{|Q|}+1
\end{align}
by the definition \eqref{eq:beta} of $\beta_Q(r)$.

Similarly, the right-hand side of \eqref{eq:GCSBK2} can be simplified as 
\begin{align}
\sum_{r \in [|U|]\setminus Q}C(A^{(r)}(U))+\sum_{q \in Q}\sum_{r=1}^{q-1}\alpha_Q(q,r)C(A^{(r)}(U))
= \frac{1}{\prod_{t=1}^{|Q|}(q_t-1)}\sum_{r=1}^{|U|}\beta_Q(r)C(A^{(r)}(U)).\label{eq:T500}
\end{align}
Substituting \eqref{eq:T400} and \eqref{eq:T500} into \eqref{eq:GCSBK2} and multiplying both sides of the inequality by $\prod_{t=1}^{|Q|}(q_t-1)$ complete the proof of Corollary~\ref{cor:GCSBK2}.


\begin{thebibliography}{99}

\bibitem{Ahl-IT00} R.~Ahlswede, N.~Cai, S.-Y.~R.~Li, and R.~W.~Yeung, ``Network information flow," \emph{IEEE Trans. Inf. Theory}, vol.~46, no.~4, pp.~1204--1216, Jul.~2000.

\bibitem{Cov-B06} T.~M.~Cover and J.~A.~Thomas, \emph{Elements of Information Theory, 2nd Edition}. John Wiley \& Sons, 2006.

\bibitem{Ere-WCCC03} E.~Erez and M.~Feder, ``Capacity region and network codes for two receivers multicast with private and common data," in \emph{Proc. Workshop Coding, Crypt. Combinatorics}, Huangshan, China, Jun.~2003.

\bibitem{For-CJM56} L.~R.~Ford and D.~R.~Fulkerson, ``Maximal flow through a network," \emph{Canadian J. Mathematics}, vol.~8, no.~3, pp.~399--404, 1956.

\bibitem{Gro-AP08} L.~Grokop and D.~N.~C.~Tse, ``Fundamental constraints on multicast capacity regions," \emph{Preprint}, 2008. Available online at \url{http://arxiv.org/abs/0809.2835}

\bibitem{Har-IT06} H.~J.~A.~Harvey, R.~Kleinberg, and A.~R.~Lehman, ``On the capacity of information networks," \emph{IEEE Trans. Inf. Theory}, vol.~52, no.~6, pp.~2345--2364, Jun.~2006.

\bibitem{Kra-JNSM06} G.~Kramer and S.~A.~Savari, ``Edge-cut bounds on network coding rates," \emph{J. Network Syst. Management}, vol.~14, no.~1, pp.~49--67, Mar.~2006.

\bibitem{Koe-ToN03} R.~Ko\'{e}tter and M.~M\'{e}dard, ``An algebraic approach to network coding," \emph{IEEE/ACM Trans. Networking}, vol.~11, no.~5, pp. 782--795, Oct.~2003.

\bibitem{Li-IT03} S.-Y.~R.~Li, R.~W.~Yeung, and N.~Cai, ``Linear network coding," \emph{IEEE Trans. Inf. Theory}, vol.~49, no.~2, pp.~371--381, Feb.~2003.

\bibitem{Nga-ICCCAS04} C.~K.~Ngai and R.~W.~Yeung, ``Multisource network coding with two sinks," in \emph{Proc. IEEE Int. Conf. Comm., Circ. Systems}, Jun.~2004, vol.~1, pp.~34--37.

\bibitem{Nga-ITW04} C.~K.~Ngai and R.~W.~Yeung, ``Network coding gain of combination networks," in \emph{Proc. IEEE Inf. Theory Workshop}, San Antonio, TX, USA, Oct.~2004, pp.~283--287.

\bibitem{Ram-CWIT05} A.~Ramamoorthy and R.~Wesel, ``The single source two terminal network with network coding," in \emph{Proc. 9th Canadian Workshop Inf. Theory}, Montr\'{e}al, Qu\'{e}bec, Canada, Jun.~2005.


\bibitem{Tia-IT11} C.~Tian, ``Latent capacity region: A case study on symmetric broadcast with common messages," \emph{IEEE Trans. Inf. Theory}, vol.~57, no.~6, pp.~3273--3285, Jun.~2011.

\bibitem{Yeu-B08} R.~W.~Yeung, \emph{Information Theory and Network Coding}. Springer, 2008.

\end{thebibliography}
\end{document}